\documentclass[a4paper,UKenglish,cleveref, autoref, thm-restate]{lipics-v2021}

\usepackage{settigs}
\usepackage{macros}

\title{Presburger arithmetic with threshold counting quantifiers is easy} 

\titlerunning{Presburger arithmetic with threshold counting quantifiers is easy} 

\author{Dmitry Chistikov}{Centre for Discrete Mathematics and its Applications (DIMAP) \&\\Department of Computer Science, University of Warwick, Coventry, UK}{d.chistikov@warwick.ac.uk}{https://orcid.org/0000-0001-9055-918X}{}

\author{Christoph Haase}{Department of Computer Science, University of Oxford, Oxford, UK}{christoph.haase@cs.ox.ac.uk}{https://orcid.org/0000-0002-5452-936X}{}

\author{Alessio Mansutti}{Department of Computer Science, University of Oxford, Oxford, UK}{alessio.mansutti@cs.ox.ac.uk}{https://orcid.org/0000-0002-1104-7299}{}

\authorrunning{D.~Chistikov, C.~Haase, A.~Mansutti} 

\Copyright{Dmitry Chistikov, Christoph Haase and Alessio Mansutti} 

\ccsdesc[500]{Theory of computation~Logic}

\keywords{Presburger arithmetic, counting quantifiers, quantifier elimination} 

\category{} 

\relatedversion{} 

\usepackage{wrapfig}
\funding{This work is part of a project that has received
  funding from the European Research Council (ERC) under the European
  Union's Horizon 2020 research and innovation programme (Grant
  agreement No.~852769, ARiAT).}

\usepackage[firstpage]{draftwatermark}

\SetWatermarkText{\hspace*{7.2in}\raisebox{-2.22in}{
    \includegraphics[scale=0.06]{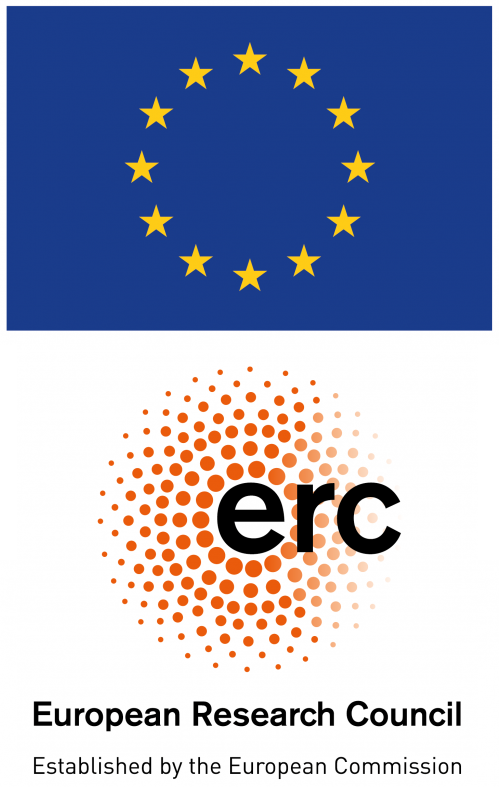}}}
\SetWatermarkAngle{0}

\acknowledgements{}

\nolinenumbers 

\hideLIPIcs  

\EventEditors{John Q. Open and Joan R. Access}
\EventNoEds{2}
\EventLongTitle{42nd Conference on Very Important Topics (CVIT 2016)}
\EventShortTitle{CVIT 2016}
\EventAcronym{CVIT}
\EventYear{2016}
\EventDate{December 24--27, 2016}
\EventLocation{Little Whinging, United Kingdom}
\EventLogo{}
\SeriesVolume{42}
\ArticleNo{23}

\begin{document}

\maketitle

\begin{abstract}
    We give a quantifier elimination procedures for the extension of
Presburger arithmetic with a unary threshold counting quantifier
$\exists^{\ge c} y$ that determines whether the number of different
$y$ satisfying some formula is at least $c\in \Nat$, where $c$ is
given in binary. Using a standard quantifier elimination procedure for
Presburger arithmetic, the resulting theory is easily seen to be
decidable in \fourexptime. Our main contribution is to develop a novel
quantifier-elimination procedure for a more general counting
quantifier that decides this theory in \threeexptime, meaning that it
is no harder to decide than standard Presburger arithmetic. As a
side result, we obtain an improved quantifier elimination
procedure for Presburger arithmetic with counting quantifiers as
studied by Schweikardt [\emph{ACM Trans.\ Comput.\ Log.}, 6(3),
pp.~634-671, 2005], and a \threeexptime quantifier-elimination
procedure for Presburger arithmetic extended with a generalised modulo
counting quantifier.

\end{abstract}

\section{Introduction}

Counting the number of solutions to an equation,
or the number of elements in a set subject to constraints,
is a fundamental problem in mathematics and computer science.
In discrete geometry, a canonical \#P-complete problem is to count the
number of integral points in polyhedra.
The celebrated algorithm due to Barvinok~\cite{Bar94} solves
this problem in polynomial time in fixed dimension.
This and other powerful insights motivate the study of algorithmic aspects
of the more general problem of counting the number of models of
formulae in \emph{Presburger arithmetic,}
the first-order theory of the integers with
addition and order, and more generally to
considering \emph{counting extensions} of this logic.

It has long been known that the decision problem for Presburger arithmetic
itself
is recursively solvable~\cite{Pre29} and that there is a quantifier elimination
procedure for Presburger arithmetic running in \threeexptime~\cite{Opp78}. 
In this article, we
study quantifier elimination procedures for 
extensions of Presburger arithmetic with \emph{counting quantifiers}.
We are primarily interested in its extension 
with a unary threshold counting quantifier $\exists^{\ge c}y$, where
$c$ is given in \emph{binary}. Given an assignment of integers to the
first-order variables $y,z_1,\ldots,z_n$, a formula $\exists^{\ge
  c}y\, \Psi(y,z_1,\ldots,z_n)$ evaluates to true whenever there are
at least $c$ different values of~$y$ satisfying $\Psi(y,z_1,\ldots,z_n)$. Note
that the number of different~$y$ may depend on the
values of the variables $z_1,\ldots,z_n$. It is easily seen that a
formula $\exists^{\ge c}y \Psi(y,z_1,\ldots,z_n)$ can equivalently be
expressed as $\exists x_1\cdots\exists x_c\, \bigwedge_{1\le i\le c}
\Psi(x_i,z_1,\ldots,z_n) \wedge \bigwedge_{1\le i<j\le c} x_i \neq
x_j$. However, since $c$ is given in binary, this translation
incurs an exponential blow-up, and thus the extension
of Presburger arithmetic with an $\exists^{\ge c}y$ counting
quantifier can only be decided in
\fourexptime using the standard quantifier elimination procedure
for Presburger arithmetic.

The main contribution of this paper is to develop a novel
quantifier-elimination procedure that enables us to solve this
decision problem in~\threeexptime, i.e., at no additional cost
compared to standard Presburger arithmetic.

\vspace{-2ex}
\subparagraph*{The counting quantifiers
$\exists^{\ge x} y$
and
$\exists^{= x} y$%
.}
The unary threshold counting quantifier we consider is a special
instance of a more general unary counting quantifier $\exists^{\ge x}
y$, which itself generalises the counting quantifier $\exists^{= x} y$
for Presburger arithmetic studied by Apelt~\cite{A66} and
Schweikardt~\cite{Schweikardt05}. Given an assignment of integers to
the first-order variables $x,z_1,\ldots,z_n$, 
a formula
$\exists^{\ge x}y\, \Psi(x,y,z_1,\ldots,z_n)$ evaluates to true
whenever the number of different $y$
satisfying $\Psi(x,y,z_1,\ldots,z_n)$ is at least the value of $x$.
The semantics of the
$\exists^{=x}y$ counting quantifier is defined analogously, but
observe that $\exists^{=x} y$ does not hold when there is an infinite
number of different $y$. Both~\cite{A66} and~\cite{Schweikardt05} show
decidability of Presburger arithmetic extended with the counting
quantifier $\exists^{=x}y$ by establishing a quantifier-elimination
procedure.

For our \threeexptime algorithm
for Presburger arithmetic extended with the $\exists^{\ge c} y$
counting quantifier, we develop a novel quantifier-elimination procedure
for the most general $\exists^{\ge x}y$ counting quantifier. While
this procedure \emph{a priori} runs in non-elementary time, we show
that it can be performed in \threeexptime when specialised to
$\exists^{\ge c}y$ counting quantifiers.



We remark that, crucially,
such counting quantifiers are always \emph{unary}, so as to
keep the logic decidable. Indeed, consider a binary counting quantifier
$\exists^{=x}(y_1,y_2)$ counting the number of different $y_1$ and
$y_2$ satisfying a formula. Then the formula ${\Phi(x,z) =
\exists^{=x}(y_1,y_2)(0\le y_1,y_2 < z)}$ holds for $x=z^2$, which in
turn allows one to define multiplication, leading to undecidability of
the resulting first-order theory.
Note that for threshold counting, in contrast, non-unary quantifiers
do not lead to undecidability. Thus, our results lead to the problem
of eliminating such quantifiers in a resource-efficient manner, which
we leave open.

\vspace{-2ex}
\subparagraph*{The counting quantifier $\exists^{(r,q)}y$.}
It is wide open whether there is an algorithm with elementary
running time that decides Presburger arithmetic extended with $\exists^{=
  x}y$ or $\exists^{\ge x}y$ counting quantifiers. Moreover, at
present, no stronger lower bounds than those established for plain
Presburger arithmetic are known~\cite{Ber80}. To shed
more light on the complexity of Presburger arithmetic with an
$\exists^{=x}y$ counting quantifier,
Habermehl and Kuske~\cite{HabermehlK15} gave a
quantifier-elimination procedure for eliminating a unary \emph{modulo
  counting} quantifier $\exists^{(r,q)}y$: here $\exists^{(r,q)}y\,
\Psi(y,z_1,\ldots,z_n)$ holds whenever the number of different $y$
satisfying $\Psi(y,z_1,\ldots,z_n)$ is congruent to $r$ modulo $q$.
An analysis of the growth of the constants and coefficients occurring
in their quantifier-elimination procedure then enables them to derive
an automata-based \threeexptime algorithm for deciding Presburger
arithmetic extended
with the $\exists^{(r,q)}y$ counting quantifier.
As a side result, we show that
our quantifier elimination procedure gives a \threeexptime upper bound for
Presburger arithmetic with a generalised
modulo counting quantifier $\exists^{(x,q)} y\,
\Psi(x,y,z_1,\ldots,z_n)$ that evaluates to true when $x$ is congruent
modulo $q$ to the number of different $y$ satisfying
$\Psi$, thereby strictly generalising the result of Habermehl and Kuske.

\vspace{-2ex}
\subparagraph*{Key techniques.} 
An advantage of our quantifier elimination procedure for the
$\exists^{\ge x}y$ and $\exists^{= x}y$ counting quantifiers is that
it avoids the introduction of additional
$\exists$- and $\forall$-quantifiers when eliminating a counting
quantifier on which
Schweikardt's procedure~\cite{Schweikardt05}
relies.
Her quantifier-elimination
procedure replaces a counting quantifier $\exists^{=x} y$ with an
equivalent quantified formula of Presburger arithmetic and requires a
full transformation into disjunctive normal form.

One key technique we employ is that, for Presburger arithmetic, it is
possible to transform any Boolean combination of inequalities 
into a ``disjoint''
disjunctive normal form in polynomial time when the number of
variables is fixed, see e.g.~\cite{Sca84,Woods15}. This enables us to
make use of a highly desirable disjunctive normal form for counting
purposes without the drawback of non-elementary growth that repeated
translation into usual
disjunctive normal form normally entails.
Another crucial ingredient of our quantifier eliminating procedure is
evaluation of counting functions on bounded segments during the
elimination process. This enables us to circumvent the introduction of
additional standard first-order quantifiers that occurs in
Schweikardt's procedure.

Putting these two ingredients together,
we obtain a procedure that, in an analogue of nondeterministic guessing,
pre-evaluates atomic predicates in the input formula.
When a quantifier is eliminated, our ``almost evaluated'' formula
is reduced to subformulae that are introduced during this ``guessing''
and evaluation of the counts.

\vspace{-2ex}
\subparagraph*{Further related work.}
The counting
quantifiers considered in this paper are derived from the so-called
\emph{H\"artig quantifier} that enables reasoning about
equicardinality between the sets defined by two formulae of
first-order logic~\cite{HKPV91}.
For Presburger arithmetic, the aforementioned undecidability
result imposes tight restrictions on how first-order variables
counting integral points can be used. Woods~\cite{Woods15} studied
properties of \emph{Presburger counting functions}: given a formula
$\Phi(x_1,\ldots,x_m,p_1,\ldots,p_n)$ of Presburger arithmetic, the
Presburger counting function associated to $\Phi$ is the function
\[
g_\Phi(p_1,\ldots,p_n) = \#\left\{ (x_1,\ldots,x_m)\in \Nat^m:
\Phi(x_1,\ldots,x_m,p_1,\ldots,p_n) \right\}\,.
\]
In~\cite{Woods15}, Woods shows that $g_\Phi$ is a so-called
\emph{piecewise quasi-polynomial}. Computing such quasi-polynomials is
of high relevance, for instance in numerous compiler optimisation
approaches, see e.g.~\cite{VSBLB07} and the references therein. More
generally, Bogart et al.~\cite{BGNW19} have recently studied
quantifier elimination procedures and counting problems for parametric
Presburger arithmetic.

\section{Presburger arithmetic with counting quantifiers}
\label{section:preliminaries}

\subparagraph*{General notation.}
The symbols~$\Zed$, $\Nat$ and~$\PNat$ denote the set of integers, natural numbers including zero, and natural numbers without zero, respectively. 
We usually use~$a,b,c,\dots$ for integer numbers, which we assume being encoded in binary.
Given $n \in \Nat$, we write $[n] \egdef \{0,\dots,n-1\}$.
We write $\card{\aset}$ for the cardinality of a set~$\aset$. 
If $\aset$~is infinite, then $\card{\aset} = \infty$,
and we postulate $n \le \infty$ for all $n \in \Zed$.

\vspace{-2ex}
\subparagraph*{Structure.}
We consider the structure $\mathcal{Z} = \langle\Zed, (c)_{c \in \Zed}, +, <, (\equiv_q)_{q \in \PNat}\rangle$, where $(c)_{c \in \Zed}$ are
constant symbols that shall be interpreted as their homographic integer numbers, the binary function symbol~$+$ is interpreted as addition on~$\Zed$, the binary relation~$<$ is interpreted as ``less than'',
and $\equiv_q$ is interpreted as the modulo relation, i.e., $a \equiv_q b$ iff $q$ divides $a-b$.

\vspace{-2ex}
\subparagraph*{Basic syntax.}
Let $\asetvar = \{\avar,\avarbis,\avarter,\dots\}$ be a countable set of first-order variables.
\emph{Linear terms}, usually denoted by~$\aterm$, $\aterm_1$, $\aterm_2$, etc., are expressions of the form $a_1\avar_1 + \dots a_d\avar_d + c$ where $\avar_1,\dots,\avar_d \in \asetvar$, $a_1,\dots,a_d,c \in \Zed$. The integer $a_i$ is the \emph{coefficient} of the variable $\avar_i$. Variables not appearing in the linear term are tacitly assumed to have a $0$ \emph{coefficient}.
A term~$\aterm$ is said to be $\avar$\emph{-free} if the coefficient of the variable~$\avar$ in~$\aterm$ is $0$.
The integer $c$ is the \emph{constant} of the linear term.
Linear terms with constant $0$ are said to be \emph{homogeneous}.

Given a term $\aterm$, the lexeme $\aterm < 0$ is understood as a \emph{linear inequality}, and $\aterm \equiv 0\bmod q$ is a \emph{modulo constraint}.
Syntactically,
Presburger arithmetic with counting quantifiers (PAC) is the closure of linear inequalities and modulo constraints under the Boolean connectives~$\land$ and~$\lnot$ (i.e.~conjunction and negation, respectively), the \emph{first-order quantifier}~$\exists\avarbis$ and the \emph{(unary) counting quantifier} $\exists^{\geq \avar}\avarbis$, where $\avar,\avarbis \in \asetvar$. We assume the two variables~$\avar$ and~$\avarbis$ appearing in a counting quantifier to be syntactically different. Formulae of PAC are denoted by~$\Aformula$, $\aformula$, $\aformulabis$, $\aformulater$, etc.
We write $\vars{\aformula}$ and $\freevars{\aformula}$ for the set of variables and free variables of~$\aformula$, respectively. For the counting quantifier, we have $\freevars{\exists^{\geq \avar}\avarbis\,\aformula} = \{\avar\} \cup (\freevars{\aformula} \setminus \{\avarbis\})$.
We say that a formula~$\aformula$ is $\avarter$\emph{-free} if $\avarter \in \asetvar$ does not occur in~$\aformula$.
Given linear terms~$\aterm$ and $\aterm'$, we write $\aformula\substitute{\aterm}{\aterm'}$ for the formula obtained from $\aformula$ by syntactically replacing every occurrence of~$\aterm$ by~$\aterm'$.

\vspace{-2ex}
\subparagraph*{Semantics.} An \emph{assignment} is a function $\aeval \colon \asetvar \to \Zed$ assigning an integer value to every variable.
As usual, we extend $\aeval$ in the standard way to a function that maps every term to an element of~$\Zed$. For instance, $\aeval(\avar + 3\avar + 2) = \aeval(x) + 3\aeval(y) + 2$.
Given a variable $\avar$ and an integer $n$, we write~$\aeval\substitute{\avar}{n}$ for the assignment obtained form $\aeval$ by updating the value of $\avar$ to $n$, i.e.~$\aeval\substitute{\avar}{n}(\avar) = n$, and for all variables $\avarbis$ distinct from $\avar$, $\aeval\substitute{\avar}{n}(\avarbis) = \aeval(\avarbis)$.
Given a formula~$\aformula$ of PAC and an assignment~$\aeval$, the satisfaction relation~$\aeval \models \aformula$ is defined as usual for linear inequalities, modulo constraints, Boolean connectives and the existential quantifier ranging over~$\Zed$. For the counting quantifier, we define
\begin{center}
  $\aeval \models \exists^{\geq \avar} \avarbis\, \aformula$
  if and only if 
  $\card{\{n \in \Zed \mid \aeval\substitute{\avarbis}{n} \models \aformula\}} \geq \aeval(\avar)$. 
\end{center}
Informally, $\exists^{\geq \avar} \avarbis\, \aformula$ is satisfied by~$\aeval$ whenever there are at least $\aeval(\avar)$ distinct values for the variable $\avarbis$ that make the formula $\aformula$ true.
A formula $\aformula$ of PAC is \emph{satisfiable} whenever there is an assignment $\aeval$ such that $\aeval \models \aformula$. Two formulae $\aformula$ and $\aformulabis$ are equivalent, written $\aformula \fequiv \aformulabis$, whenever they are satisfied by the same set of assignments.

\vspace{-2ex}
\subparagraph*{Syntactic abbreviations.} We define~$\false\, \egdef 0 < 0$ and~$\true \egdef {\lnot \false}$. The Boolean connectives~$\lor$, $\rightarrow$ and $\fequiv$ and the universal first-order quantifier~$\forall$ are derived as usual, and so are the (in)equalities $<$, $\leq$, $=$, $\geq$, and $>$, between terms. For instance, $\aterm_1 < \aterm_2$ corresponds to $\aterm_1 - \aterm_2 < 0$, where we tacitly manipulate $\aterm_1 - \aterm_2$ with standard operation of linear arithmetic in order to obtain an equivalent term.
Given two terms~$\aterm_1$ and $\aterm_2$, and~$q \in \PNat$, we write $\aterm_1 \equiv_{q} \aterm_2$ for the modulo constraint $\aterm_1 - \aterm_2 \equiv 0 \bmod q$. For a variable~$\avar \in \asetvar$ and~$r \in [q]$, we call $\avar \equiv_q r$ a \emph{simple} modulo constraint. All modulo constraints introduced by our quantifier elimination procedure given in~\Cref{section:quantifier-elimination} are simple.

We now introduce the counting quantifiers~$\exists^{\geq c}\avarbis$, $\exists^{=\avar}\avarbis$ and $\exists^{(\avar,q)}\avarbis$, where $\avar,\avarbis \in \asetvar$, $c \in \Zed$ and $q \geq 1$, and both $c$ and $q$ are encoded in binary.
Let $\aeval$ be an assignment.
Informally, $\aeval \models \exists^{\geq c} \avarbis\,\aformula$ if and only if there are at least $c$ values for the variable~$\avarbis$ that make $\aformula$ true.
Similarly, $\aeval \models \exists^{= \avar} \avarbis\,\aformula$ 
if and only if there are exactly $\aeval(\avar)$ values for the variable $\avarbis$ that make~$\aformula$ true. 
Finally, $\aeval \models \exists^{(\avar,q)} \avarbis\,\aformula$ 
if and only if the number of values for the variable $\avarbis$ that make $\aformula$ true is congruent to~$\aeval(\avar)$ modulo $q$.
The formal definition of these three counting quantifiers is given below, where $\avarter$ is a variable not occurring in $\freevars{\aformula}$,
\begin{align*}
    \exists^{\geq c} \avarbis\,\aformula &\ \egdef \ \exists \avarter\,(\avarter = c \land \exists^{\geq \avarter} \avarbis\,\aformula),\\
    \exists^{= \avar} \avarbis\,\aformula &\ \egdef \ (\exists^{\geq \avar}\avarbis\,\aformula) \land \lnot \exists \avarter\,(\avarter = \avar+1 \land \exists^{\geq \avarter}\,\aformula),\\
    \exists^{(\avar,q)}\avarbis\, \aformula &\ \egdef \ \exists\avarter\, (\avarter \equiv_q \avar \land \exists^{=\avarter} \avarbis\, \aformula).
\end{align*}
The counting quantifier~$\exists^{=\avar}\avarbis$ is the counting quantifier considered in~\cite{Schweikardt05}, whereas the  
\emph{modulo counting quantifier}~$\exists^{(\avar,q)}\avarbis$ is a generalisation of the quantifier~$\exists^{(r,q)}\avarbis$ introduced in~\cite{HabermehlK15}, where $r$ is a fixed natural number instead of a variable. More precisely, $\exists^{(r,q)}\avarbis\,\aformula$ is equivalent to $\exists \avar\,(\avar \equiv_q r \land \exists^{(\avar,q)}\avarbis\,\aformula)$.
Finally, notice that the formula $\exists \avarbis\, \aformula$ is equivalent to $\exists^{\geq 1} \avarbis\, \aformula$. This fact enables us to eliminate a standard first-order existential quantifier by slightly tweaking the quantifier-elimination procedure for~$\exists^{\geq c}\avarbis\, \aformula$.

\vspace{-2ex}
\subparagraph*{Parameters of formulae.}
Following Oppen~\cite{Opp78} and Weispfenning~\cite{Weispfenning90}, we establish bounds on the absolute value of the 
variable assignments that suffice for deciding satisfiability of~$\aformula$.
To this end, we introduce a set of parameters of formulae of PAC:
\begin{itemize}
\item $\abs{\aformula}$ denotes the length of the formula $\aformula$, i.e., the number of symbols to write down $\varphi$, with numbers encoded in binary,
\item
$\linterms{\aformula}$ is the set of all linear terms $\aterm$ that appear in a linear inequality~$\aterm < 0$ of~$\aformula$ (recall that $\aterm_1 < \aterm_2$ is syntactic sugar for~$\aterm_1 - \aterm_2 < 0$),
\item
  $\homterms{\aformula}$ is the set of homogeneous linear terms obtained from all terms in $\linterms{\aformula}$ by setting their constants to $0$, and
\item
$\fmod{\aformula}$ is the set of all moduli~$q \in \Nat$ appearing in a modulo constraint $\aterm_1 \equiv_q \aterm_2$ of $\aformula$.
We always assume~$1 \in \fmod{\aformula}$, even if $\aformula$ has no modulo constraints.
\end{itemize}
For $\aset \subseteq \Zed$ finite set, we write $\norminf{\aset} = \max\{\abs{n} \mid n \in \aset\}$ for the \emph{absolute-value norm} of~$\aset$.
For a term~$\aterm$, $\norminf{t}$ it the maximum coefficient or constant appearing in $t$, in absolute value. For a set of terms $T$, $\norminf{T} \egdef \max\{\norminf{t} \mid t \in T\}$. For a formula~$\aformula$, we define $\norminf{\aformula}$ as~$\norminf{\linterms{\aformula} \cup \fmod{\aformula}}$.

\section{A quantifier elimination procedure for~unary counting quantifiers}
\label{section:quantifier-elimination}
In this section, we develop a quantifier elimination procedure (QE procedure) for the counting quantifier~$\exists^{\geq \avar}\avarbis$ that allows us to establish the following result.

\begin{theorem}
    \label{theorem:counting-quantifier-elimination}
    Let $\aformula$ be quantifier-free. 
    Then~$\exists^{\geq \avar}\avarbis\, \aformula$ is equivalent to a Boolean combination of linear inequalities and simple modulo constraints.
\end{theorem}

Our QE procedure perform a series of formula manipulations that we divide into five steps. At the end of the $i$-th step, the procedure produces a formula
$\Aformula_i$ equivalent to the original formula $\exists^{\geq \avar}\avarbis\,\aformula$. 
Ultimately, $\Aformula_5$ is a Boolean combination of inequalities and simple modulo constraints allowing us to establish~\Cref{theorem:counting-quantifier-elimination}.
In this section, we present the procedure and briefly discuss its correctness, leaving the computational analysis of parameters~$\linterms{\Aformula_5}$, $\homterms{\Aformula_5}$ and~$\fmod{\Aformula_5}$
to subsequent sections.

\subparagraph*{Step \Roman{eliproc}: Normalise the coefficients of the variable~$\avarbis$.}
Given the input formula~$\Aformula_0 = \exists^{\geq \avar}\avarbis\, \aformula$,
the first step of the procedure is a standard step for QE procedures for Presburger arithmetic. It produces an equivalent formula $\Aformula_1$ in which all non-zero coefficients of~$y$ appearing in a linear term are normalised to $1$ or $-1$. 
For simplicity, we first translate every modulo constraint in~$\aformula$ into simple modulo constraints, by relying on the lemma below.%
\begin{restatable}{lemma}{LemmaMakeModuloSimple}
    \label{lemma:make-modulo-simple}
    Every modulo constraint~$\aterm \equiv_q 0$ is equivalent to a disjunction~$\aformulabis$ of simple modulo constraints such that~$\vars{\aformulabis} \subseteq \vars{\aterm \equiv_q 0}$ and $\fmod{\aformulabis} = \{q\}$.
\end{restatable}
Here is the first step of the procedure:
\begin{itemize}
    \setlength{\itemsep}{2pt}
    \item Using~\Cref{lemma:make-modulo-simple}, translate every modulo constraint in~$\aformula$ into simple modulo constraints.
    \item Let $k$ be the \emph{lcm} of the absolute values of all coefficients of~$y$ appearing in $\homterms{\aformula}$.
    \item Let $\aformula'$ be the formula obtained from~$\aformula$ by applying the following three rewrite rules to each linear inequality and simple modulo constraint in which $\avarbis$~appears, where $\aterm$ is a term, $q\geq1$ and $r \in [q]$:
      \begin{itemize}
      \item
        $a\avarbis + \aterm < 0 \ \longrightarrow \  k\avarbis + (k/a)\cdot \aterm < 0$, \quad if $a > 0$,
      \item
        $a\avarbis + \aterm < 0 \ \longrightarrow \ -k\avarbis - (k/a)\cdot \aterm < 0$, \quad if $a < 0$, and
      \item
        $\avarbis \equiv_q r \ \longrightarrow \ k\avarbis \equiv_{kq} kr$.
      \end{itemize}
    \item Let $\Aformula_1 \egdef \exists^{\geq \avar}\avarbis\, (\avarbis \equiv_k 0 \land \aformula'\substitute{k\avarbis}{\avarbis})$.
\end{itemize}

\begin{restatable}{claim}{claimPsiOne} 
    \label{claim:psi1}
    $\Aformula_0 \fequiv \Aformula_1$, and in~$\Aformula_1$, all non-zero coefficients of~$y$ are either~$1$ or~$-1$.
\end{restatable}

\vspace{-2ex}
\stepcounter{eliproc}
\subparagraph*{Step \Roman{eliproc}: Subdivide the formula according to term orderings and residue classes.}
We define an \emph{ordering of $n$~linear terms} to be a formula of the form
\begin{equation}
\label{eq:ordering}
        (\aterm_1 \inormod_1 \aterm_{2}) \land
        (\aterm_2 \inormod_2 \aterm_{3}) \land \dots \land
        (\aterm_{n-1} \inormod_{n-1} \aterm_n),
\end{equation}
where $\{t_1, \dots, t_n\}$ is the set being ordered and
$\{\inormod_1,\dots,\inormod_{n-1}\} \subseteq \{<, = \}$.

\begin{restatable}{lemma}{lemmaBoundNumberOfOrderings}
    \label{lemma:bound-number-of-orderings}
    There is an algorithm that,
    given a set $T$ of $n$ linear terms over~$d$~variables, 
    computes in time 
    $n^{\mathcal O (d)} \log \norminf{T}^{\mathcal O (1)}$ a set 
    $\{\ordering_1,\dots,\ordering_o\}$ of orderings for~$T$ such that 
    (I)~$\bigvee_{i \in [1,o]} \ordering_i$ is a tautology,
    (II) for every $i \neq j$ in $[1,o]$, $\ordering_{i} \land \ordering_{j}$ is unsatisfiable, 
    and (III) $o = \mathcal O(n^{2 d})$.
\end{restatable}


Our QE procedure manipulates $\Aformula_1$ as follows:
\begin{itemize}
    \setlength{\itemsep}{2pt}
    \item Let $T$ be the set of all $y$-free terms~$t$ such that $t$, $y - t$ or $-y + t$ belong to~$\linterms{\Aformula_1}$.
    \item Using~\Cref{lemma:bound-number-of-orderings}, build a disjunction of~orderings $\aformulabis_{\text{ord}} \egdef \bigvee_{i \in [1,o]} \ordering_i$ for the terms~${T \cup \{0\}}$.
    \item Let $Z = \vars{\aformula}$ and $m = \lcm(\fmod{\Aformula_1})$.
    \item
    For each $i \in [1,o]$ and every $r \colon Z \to [m]$, 
    let~$\Gamma_{i,r} \egdef \ordering_i \land (\bigwedge_{z \in Z} z \equiv_m r(z))$.
    \item Let $\Aformula_2 \egdef \bigvee_{i \in [1,o]}\bigvee_{r\colon Z\to [m]}\left(\Gamma_{i,r} \land \Aformula_1 \right)$.
\end{itemize}

\begin{restatable}{claim}{claimPsiTwo}
    \label{claim:psi2}
    $\Aformula_1 \fequiv \Aformula_2$.
\end{restatable}

In Steps~III to V of the procedure,
we focus on each disjunct separately,
iterating over all pairs of~$i \in [1,o]$ and $r\colon Z \to [m]$.

\vspace{-2ex}
\stepcounter{eliproc}
\subparagraph*{Step \Roman{eliproc}: Split the range of~$y$ into segments.}
Recall that~$\Aformula_1 = \exists^{\geq \avar}\avarbis\,\aformulabis$, where  $\aformulabis$ is some Boolean combination of inequalities and modulo constraints with variables from~$\vars{\aformula}$, in which the non-zero coefficients of~$y$ are either~$1$ or~$-1$ (by~\Cref{claim:psi1}).
Let $\aterm_1',\dots,\aterm_\ell'$ be all of the terms $T \cup \{0\}$
that the formula~$\ordering_i$ asserts pairwise non-equal, taken in the ascending order.
In other words, we obtain $\aterm_1',\ldots,\aterm_\ell'$
by removing from the sequence $\aterm_1, \dots, \aterm_n$ in~\Cref{eq:ordering}
all terms~$t_{j+1}$ for which $\inormod_{j}$ is $=$.
Let $\segments(\avarbis,\ordering_i)$ be the set of formulae
\begin{center}
    $\bigl\{\,\avarbis < \aterm_1'$, \ $y = \aterm_1'$, \ $(\aterm_1' < \avarbis \land \avarbis < \aterm_2')$, \ $y = \aterm_2'$, \ $\dots$, \ $(\aterm_{\ell-1}' < y \land y < \aterm_\ell')$, \ ${y = \aterm_\ell'}$, \ $\aterm_\ell' < y\,\bigr\}$.
\end{center}
We have $\card({\segments(\avarbis,\ordering_i)}) = 2\ell + 1$.
Given $\aformulafour \in \segments(\avarbis,\ordering_i)$, the formula $\ordering_i \land \aformulafour$ imparts a linear ordering on the terms $T \cup \{0,\avarbis\}$.
This enables us to ``almost evaluate'' the formula $\aformulabis$:

\begin{restatable}{claim}{claimPsiThreeSimpleMod}
    \label{claim:psi3-simpl-mod}
    For every $\aformulafour \in \segments(\avarbis,\ordering_i)$, there is a Boolean combination~$\aformulabis_{\aformulafour}^{i,r}$ of simple modulo constraints s.t.~$\vars{\aformulabis_{\aformulafour}^{i,r}} = \{\avarbis\}$, $\fmod{\aformulabis_{\aformulafour}^{i,r}} \subseteq \fmod{\aformulabis}$ and $\Gamma_{i,r} \land \aformulafour \land \aformulabis$ 
        \ $\fequiv$ \ 
        $\Gamma_{i,r} \land \aformulafour \land \aformulabis_{\aformulafour}^{i,r}$.
\end{restatable}

The procedure continues as follows:
\begin{itemize}
    \item Let $\segments(y,\ordering_i) = \{\aformulafour_0,\dots,\aformulafour_{2\ell}\}$ and, for every $j \in [0,2\ell]$, take $\aformulabis^{i,r}_{\aformulafour_j}$ from~\Cref{claim:psi3-simpl-mod}.
    \item Let $\Aformula_3^{i,r} \egdef \exists \avar_0\dots\exists\avar_{2\ell} \left(
            \avar \leq \avar_0 + \dots + \avar_{2\ell} \land \bigwedge_{j \in [0,2\ell]} \exists^{\geq \avar_j} \avarbis (\aformulafour_j \land \aformulabis^{i,r}_{\aformulafour_j})\right)$.
    \item Let $\Aformula_3 \egdef \bigvee_{i \in [1,o]}\bigvee_{r\colon Z\to [m]} (\Gamma_{i,r} \land \Aformula_3^{i,r})$.
\end{itemize}

\begin{restatable}{claim}{claimPsiThree}
    \label{claim:psi3}
    $\Aformula_2 \fequiv \Aformula_3$.
\end{restatable}

\vspace{-2ex}
\stepcounter{eliproc}
\subparagraph*{Step \Roman{eliproc}: Compute the number of solutions for each segment.}
We next aim to eliminate the counting quantifiers introduced in Step~III
in the sub-formulae $\exists^{\geq \avar_j} \avarbis (\aformulafour_j \land \aformulabis^{i,r}_{\aformulafour_j})$.
We go over each $\aformulafour \in \segments(y,\ordering_i)$, and consider three cases
depending on whether it specifies (syntactically) an infinite interval, a finite segment,
or a single value for~$y$.

Notice that $r$ is in fact an assignment to variables,
so $r(t) \in \Zed$ is well-defined for every term~$t$ with free variables~$Z$.
Compute the following numbers for $j \in [1,\ell]$:
\begin{itemize}
\item
    $c_j$ is $1$ if
    the assignment $y \mapsto r(t_j')$ satisfies 
    $\aformulabis_{\aformulafour}^{i,r}$ where $\aformulafour = ( y = \aterm_j' )$ and $0$ otherwise.
\end{itemize}
For $j \in [2, \ell]$:
\begin{itemize}
    \setlength{\itemsep}{2pt}
    \item
    $p_j \in [0, m]$ is the number of~$y \in [m]$ satisfying  $\aformulabis_{\aformulafour}^{i,r}$,
    \item
    $\underline{u}_j = (r(t_{j-1}') \bmod m)$ and $\overline{u}_j$ is the smallest integer congruent to $r(t_j')$ mod~$m$ and~$>\underline{u}_j$,
    \item
    $r_j' \in [0, m]$ is the number of $y \in [\underline{u}_j+1,\overline{u}_j-1]$ satisfying~$\aformulabis_{\aformulafour}^{i,r}$,
    \item
    $r_j \in [-m^2, m^2]$ and $r_j$ is $- p_j \cdot (\overline{u}_j - \underline{u}_j) + m \cdot r_j'$.
\end{itemize}

%

\begin{restatable}{lemma}{lemmaSharpSat}
\label{lemma:sharp-sat}
Given a formula~$\aformulabis_{\aformulafour}^{i,r}$
and $m, \underline{u}_j, \overline{u}_j$,
the numbers~$p_j$ and $r_j'$ can be computed in~\textup{\sharpP},
or by a deterministic algorithm with running time
$\mathcal O(m \cdot \abs{\aformulabis_{\aformulafour}^{i,r}})$.
\end{restatable}

The numbers $c_j$, $p_j$, $r_j$ determine, for each formula
$\aformulafour \in \segments(y,\ordering_i)$, how many assignments to the variable~$y$
satisfy the formula~$\aformulabis_{\aformulafour}^{i,r}$
in the conjunction $\Gamma_{i,r} \land \aformulafour \land \aformulabis_{\aformulafour}^{i,r}$.
Intuitively, this is
$c_j$ for $\aformulafour$ of the form $y = \aterm_j'$, and
$(p_j(t_{j}' - t_{j-1}') + r_j) / m$
for $\aformulafour$ of the form $\aterm_{j-1}' < \avarbis \land \avarbis < \aterm_{j}'$.
We say ``intuitively'' here, because in the latter case the expression above depends on
other variables so is not, strictly speaking, a number.
The following claims formalise this.


\begin{restatable}{claim}{claimPsiFourInfSol}
\label{claim:psi4-inf-sol}
    Let $\aformulafour \in \{\avarbis < t_1',\ t_\ell' < \avarbis\}$.
    If $\exists \avarbis\, (\aformulafour \land \aformulabis_{\aformulafour}^{i,r})$ is satisfiable, then $\Gamma_{i,r} \land \Aformula^{i,r}_3$ $\fequiv$ $\Gamma_{i,r}$.
\end{restatable}

\begin{restatable}{claim}{claimPsiFourRemoveCqEq}
\label{claim:psi4-remove-cq-eq}
    Let $\aformulafour$ be the formula $y = \aterm_j'$
    for some $j \in [1,\ell]$ and
    let $\avarter$ be a fresh variable.
    Then
    $\Gamma_{i,r} \land \exists^{\geq \avarter}\avarbis\,(\aformulafour \land \aformulabis_{\aformulafour}^{i,r})$ \ $\fequiv$ \ $\Gamma_{i,r} \land \avarter \leq c_j$.
\end{restatable}

\begin{restatable}{claim}{claimPsiFourRemoveCqIneq}
\label{claim:psi4-remove-cq-ineq}
    Let $\aformulafour$ be the formula $\aterm_{j-1}' < \avarbis \land \avarbis < \aterm_{j}'$
    for some $j \in [2,\ell]$ and
    let $\avarter$ be a fresh variable.
    Then
    $\Gamma_{i,r} \land \exists^{\geq \avarter}\avarbis\,(\aformulafour \land \aformulabis^{i,r}_{\aformulafour})$ \ $\fequiv$ \ $\Gamma_{i,r} \land m\avarter \leq p_j(t_{j}' - t_{j-1}') + r_j$.
\end{restatable}

The procedure replaces each disjunct 
of $\Aformula_3$ with a new formula as follows:
\begin{itemize}
    \item
    Let $\Aformula_4^{i,r} \egdef \true$ 
    if $\exists \avarbis\, (\aformulafour \land \aformulabis_{\aformulafour}^{i,r})$ is satisfiable for some $\aformulafour \in \{\avarbis < t_1',\ t_\ell' < \avarbis\}$;
    otherwise\\let $\Aformula_4^{i,r} \egdef
    \exists \avar_2\dots\exists\avar_{\ell} \left(
        \avar \leq \sum_{j=2}^{\ell}\avar_j + \sum_{j=1}^{\ell}c_j \land 
        \bigwedge_{j \in [2,\ell]} m\avar_{j} \leq p_j(t_{j}' - t_{j-1}') + r_j
        \right)$.
    \item Let $\Aformula_4 \egdef \bigvee_{i \in [1,o]}\bigvee_{r\colon Z\to [m]} (\Gamma_{i,r} \land \Aformula_4^{i,r})$.
\end{itemize}

\begin{restatable}{claim}{claimPsiFour}
\label{claim:psi4}
    $\Aformula_3$ $\fequiv$ $\Aformula_4$.
\end{restatable}

\vspace{-2ex}
\stepcounter{eliproc}
\subparagraph*{Step \Roman{eliproc}: Sum up the numbers of solutions.}
It remains to get rid of the variables~$x_i$ introduced earlier.
For each
disjunct~$\Gamma_{i,r} \land \Aformula_4^{i,r}$ of $\Aformula_4$,
we use the notation from Step~IV.
\begin{itemize}
    \setlength{\itemsep}{2pt}
\item
Let~$\Aformula_5^{i,r} \egdef \true$
if~$\Aformula_4^{i,r} = \true$;
otherwise\\[2pt]
    let
    $\Aformula_5^{i,r} \egdef
    m\avar \leq \sum_{j=2}^{\ell}( p_j(t_{j}' - t_{j-1}') + r_j ) + m \cdot \sum_{j=1}^{\ell} c_j$.
\item Let $\Aformula_5 \egdef \bigvee_{i \in [1,o]}\bigvee_{r\colon Z\to [m]} (\Gamma_{i,r} \land \Aformula_5^{i,r})$.
\end{itemize}
The procedure terminates with~$\Aformula_5$ as output.
The following claim implies~\Cref{theorem:counting-quantifier-elimination}.

\begin{restatable}{claim}{claimPsiFive}
\label{claim:psi5}
$\Aformula_4$ $\fequiv$ $\Aformula_5$. The formula $\Aformula_5$ is quantifier-free.
\end{restatable}


\section{Discussion and summary of results, and roadmap}
\label{s:summary}

The QE
procedure
for a single counting quantifier~$\exists^{\geq \avar}\avarbis$
from Section~\ref{section:quantifier-elimination}
forms the basis of our results.
In this section we discuss its use and lay out its applications.

\vspace{-2ex}
\subparagraph*{Analysis of the procedure.}
The next lemma tells us how fast formulae and their
parameters grow in our QE procedure.

\begin{lemma}
    \label{lemma:bound-quantifier-elimination}
    Let the formula $\Aformula_5$ be obtained by applying
    the quantifier elimination procedure
    from Section~\ref{section:quantifier-elimination}
    to a formula~$\exists^{\geq \avarbis}\avar\,\aformula$, where $\aformula$ is quantifier-free and $\card{\vars{\aformula}} = d$.
    Then:
    \begin{gather*}
        \fmod{\Aformula_5} = \{m\} \quad\text{with\ } m = k \cdot \lcm{(\fmod{\aformula})} \text{\ and\ }k \le \norminf{\homterms{\aformula}}^{\card{\homterms{\aformula}}}, \\
        \begin{aligned}
        \card{\linterms{\Aformula_5}} &\le N^{O(d)}, &
        \norminf{\linterms{\Aformula_5}} &\le \BigO{ N } \cdot \norminf{\linterms{\aformula}}, \\
        \card{\homterms{\Aformula_5}} &\le N^{O(d)}, &
        \norminf{\homterms{\Aformula_5}} &\le \BigO{ N } \cdot \norminf{\homterms{\aformula}},
        \quad\text{where\ }
        N = m^2 \cdot \card{\linterms{\aformula}}.
        \end{aligned}
    \end{gather*}
\end{lemma}

A trivial consequence of~\Cref{lemma:bound-quantifier-elimination}
is that
the QE procedure
from Section~\ref{section:quantifier-elimination}
gives an algorithm for deciding a formula $\varphi$ of
Presburger arithmetic with counting quantifiers~$\exists^{\geq z}\avarbis$
in time $2^{\iddots{\raisebox{4pt}{\,\scriptsize$2$}}}$, where the height of the tower is at most $\BigO{\abs{\aformula}}$.

\begin{remark}
\label{remark:equality-quantifier-elimination}
With minor changes to the procedure, one can
eliminate quantifiers~$\exists^{=\avar}\avarbis$ in addition to~$\exists^{\geq \avar}\avarbis$,
with the same complexity bounds as in~\Cref{lemma:bound-quantifier-elimination}.
Because of space constraints, this is only briefly described in~\Cref{subsection:elimination-modulo-quantifiers} and
further details are relegated to~\Cref{subsection:exists-equal-elim}.
\end{remark}

Let us pinpoint where the non-elementary blow-up appears
if the procedure is applied multiple times to eliminate
\emph{all} quantifiers from a formula.
%
%
Putting together the upper bounds and equations
given by~\Cref{lemma:bound-quantifier-elimination}
for 
$\card{\homterms{\Aformula_5}}$,
$N$,
$m$, and
$k$,
we observe that
the upper bound for~$\card{\homterms{\Aformula_5}}$
is exponential in~$\card{\homterms{\aformula}}$.
This means that more fine-grained bounds are necessary for
decision procedures with elementary complexity, i.e.,
with running time bounded from above by a $k$-fold exponential in
the size of the input formula.

Tracing the exponential dependence
of~$\card{\homterms{\Aformula_5}}$
on~$\card{\homterms{\aformula}}$ back to the QE procedure,
one can see that the quantity~$k$ from~\Cref{lemma:bound-quantifier-elimination}
stems from computing the least common multiple of the coefficients
at~$y$ in Step~I of the procedure.
Each of them can be as big as~$\norminf{\homterms{\aformula}}$,
and there can be $\card{\homterms{\aformula}}$-many of them.
Unfortunately, there does not appear to be a stronger upper bound
on the magnitude of their common multiple, even in subsequent rounds
of the QE procedure.
Indeed, $y$-free terms $t_1', \ldots, t_\ell'$ in the remaining variables
do not only get subtracted from one another
in $\Aformula_5^{i,r}$, 
but also get multiplied by factors~$p_j$ as they are in formulae~$\Aformula_4^{i,r}$.
These factors, computed at the beginning of Step~IV,
represent the limit density of
suitable assignments for $y$ in the intervals~%
$\aterm_{j-1}' < \avarbis < \aterm_{j}'$
that are long enough.
As such, they are model counts
of univariate quantifier-free formulae~$\aformulabis_{\aformulafour}^{i,r}$,
so a~priori nothing prevents many different factors~$p_j$
from taking different values in the range~$[0, m]$
and contributing to a big least common multiple in the next QE round.

\vspace{-2ex}
\subparagraph*{\threeexptime decision procedures.}
We will now explain how this growth of parameters can be countered for more restricted quantifiers,
arriving at a \threeexptime quantifier elimination and decision procedures.
This analysis relies on and extends~\Cref{lemma:bound-quantifier-elimination}.
The following theorem is our main result.

\begin{restatable}{theorem}{TheoremThresholdThreeExp}
\label{theorem:threshold}
There is a \threeexptime quantifier elimination procedure for
Presburger arithmetic with threshold counting quantifiers~$\exists^{\geq c}\avarbis$.
\end{restatable}

In essence, the procedure of Theorem~\ref{theorem:threshold}
is our main QE procedure from~\Cref{section:quantifier-elimination}
that treats the quantifier~$\exists^{\geq c}\avarbis$
as if it were~$\exists^{\geq z}\avarbis$.
After substituting~$c$ for~$z$ at the end, we are able to
improve the bound on~$\card{\homterms{\Aformula_5}}$
from~\Cref{lemma:bound-quantifier-elimination}.
This results in an elementary decision procedure.
We discuss details in~\Cref{section:elimination-treshold-quantifiers}.


Similarly to the case of quantifiers~$\exists^{=\avar}\avarbis$ mentioned in~\Cref{remark:equality-quantifier-elimination},
with very minor changes to the procedure
from Section~\ref{section:quantifier-elimination}
one can
eliminate quantifiers~$\exists^{(\avar,q)}\avarbis$ too,
with the same complexity bounds as in~\Cref{lemma:bound-quantifier-elimination}.

\begin{restatable}{theorem}{TheoremModuloThreeExp}
\label{theorem:modulo}
There is a \threeexptime quantifier elimination procedure for
Presburger arithmetic with modulo counting quantifiers
$\exists^{(\avar,q)}\avarbis$.
\end{restatable}

A proof outline for Theorem~\ref{theorem:modulo} is given in~\Cref{subsection:elimination-modulo-quantifiers}.

\section{Eliminating threshold counting quantifiers}
\label{section:elimination-treshold-quantifiers}

In this section, we extend the quantifier elimination procedure of~\Cref{section:quantifier-elimination}
in order to directly deal with the threshold counting quantifiers~$\exists^{\geq c} \avarbis$.
Afterwards, we provide the complexity analysis of the quantifier elimination procedure.

Consider a formula~$\Aformula_0 = \exists^{\geq c}\avarbis\,\aformula$ with $\varphi$ quantifier free.
By definition,~$\Aformula_0$ is equivalent to the formula~$\exists \avarter\,(\avarter = c \land \exists^{\geq \avarter} \avarbis\,\aformula)$ for some variable~$\avarter$ not occurring in~$\aformula$. 
In order to eliminate the threshold counting quantifier~$\exists^{\geq c}\avarbis$, 
we first perform the quantifier elimination procedure described in~\Cref{section:quantifier-elimination} on input~$\exists^{\geq \avarter} \avarbis\,\aformula$, obtaining the formula $\Aformula_5$.
We then eliminate the existential quantifier~$\exists \avarter$ from the formula~$\exists \avarter\,(\avarter = c \land \Aformula_5)$ by relying on the ad-hoc procedure we now describe.
As explained in~\Cref{s:summary}, the prefix ``$\exists \avarter\,(\avarter = c \,\land$'' allows us to drastically simplify the set of homogeneous terms in~$\Aformula_5$, leading to~\threeexptime.

\subparagraph*{Dealing with threshold quantifiers in a single step.}
With~$\Aformula_0$ and~$\Aformula_5$ defined as above,
we have
$\Aformula_0$ $\fequiv$ $\exists \avarter\,(\avarter = c \land \Aformula_5)$ $\fequiv$ $\Aformula_5\substitute{\avarter}{c}$.
Recall that $\Aformula_5$ is defined as 
\begin{center}
$\Aformula_5 \egdef \bigvee_{i \in [1,o]}\bigvee_{r\colon Z\to [m]} (\Gamma_{i,r} \land \Aformula_5^{i,r})$.
\end{center}
Here, $Z = \vars{\aformula}$, $m = \lcm(\fmod{\Aformula_1})$ (defined as in the Step II of the procedure) and~$\Gamma_{i,r}$ is a conjunction of inequalities and simple modulo constraints with variables from~$Z$ 
(hence,~\mbox{$\avarter$-free}). 
Therefore,~$\Gamma_{i,r}\substitute{\avarter}{c} = \Gamma_{i,r}$.
Moreover,~$\Aformula_5^{i,r}$ is either $\true$ or a formula of the form 
\begin{equation}
    \label{equation:5ir-form}
    m\avarter \leq p_2(t_{2}' - t_{1}') + r_2 + \dots + p_\ell(t_{\ell}' - t_{\ell-1}') + r_\ell + m(c_1 + \dots + c_\ell).
\end{equation}
where the terms $\aterm_1',\dots,\aterm_{\ell}'$ are from~$T \cup \{0\}$ (with~$T$ defined as in Step II of~\Cref{section:quantifier-elimination}), and thus written with variables from~$Z$. 
Therefore, the following property holds:
\begin{claim}
    \label{claim:zed-almost-free}
    In~$\Aformula_5$, $\avarter$ only appears on the left hand side of 
    inequalities of the form~\eqref{equation:5ir-form}. 
\end{claim}

We manipulate the disjuncts of~$\Aformula_5\substitute{\avarter}{c}$ separately.
Fix $i \in [1,o]$ and~${r \colon Z \to [m]}$. We define a formula equivalent to~$\Aformula_5^{i,r}\substitute{\avarter}{c}$ by relying on the following lemma, where~${d = \vars{\Aformula_5^{i,r}}}$. 

\begin{restatable}{lemma}{LemmaSimplifyingThreshold}
    \label{lemma:simplifying-threshold}
    Consider $\Aformula_5^{i,r}$ as in~\eqref{equation:5ir-form}. Let $e = m(c - \sum_{j = 1}^\ell c_j) - \sum_{j = 2}^{\ell} r_j$.
    It is possible to compute in time $(e + \ell)^{\BigO{d}} \log(c\cdot\norminf{\Aformula_5^{i,r}})^{\BigO{1}}$ 
    a formula~${\aformulater_{i,r} = \bigvee_{(i_2,\dots,i_\ell) \in I} \bigwedge_{j \in [2,\ell]} \aterm_j' - \aterm_{j-1}' \geq i_j}$ such that (1) $I \subseteq [0,e]^\ell$, (2) $\card{I} \leq \BigO{(e+\ell)^{2d}}$, and (3) 
    $\Gamma_{i,r} \land \Aformula_5^{i,r}\substitute{\avarter}{c}$ \,$\fequiv$ \,$\Gamma_{i,r} \land \aformulater_{i,r}$.
\end{restatable}

\noindent
To prove~\Cref{lemma:simplifying-threshold}, we apply~\Cref{lemma:bound-number-of-orderings} on the set of terms ${\{\aterm_{j}' - \aterm_{j-1}' \mid j \in [2,\ell]\} \cup [0,e]}$, and manipulate the resulting tautology to filter out orderings that do not satisfy~$\Aformula_5^{i,r}\substitute{\avarter}{c}$.

The QE procedure proceeds as follows.

\begin{itemize}
\item For every $i \in [1,o]$ and $r \colon Z \to [m]$,
    \begin{itemize}
    \item if $\Aformula_5^{i,r} = \true$, then let $\Aformula_6^{i,r}  \, \egdef \, \true$, 
    \item else let $\Aformula_6^{i,r} \, \egdef \, \aformulater_{i,r}$, according to $\text{\Cref{lemma:simplifying-threshold}}$.
    \end{itemize}
\item Let $\Aformula_6^c \, \egdef \, \bigvee_{i \in [1,o]}\bigvee_{r\colon Z\to [m]} (\Gamma_{i,r} \land \Aformula_6^{i,r})$. 
\end{itemize}

After defining~$\Aformula_6^c$, the procedure ends.
Notice that the inequalities~$\aterm_j' - \aterm_{j-1}' \geq i_j$ that replace the inequalities given in~\eqref{equation:5ir-form} are such that $\aterm_{j-1}',\aterm_{j}' \in T \cup \{0\}$. This leads to a better bound on the size of the set~$\homterms{\Aformula_6^c}$ (more precisely, quadratic on~$\card{\homterms{\aformula}}$), which 
ultimately enables us to establish the~\threeexptime membership of Presburger arithmetic with threshold counting quantifiers.

\begin{restatable}{claim}{ClaimThresholdQEProcedure}
    $\Aformula_0 \fequiv \Aformula_6^c$. The formula $\Aformula_6^c$ is quantifier-free.
\end{restatable}

\subsection*{Proof idea for Theorem~\ref{theorem:threshold}}

The key role in the analysis is played by the following lemma.

\begin{restatable}{lemma}{LemmaBoundQuantifierThresholdSimplified}
    \label{lemma:bound-quantifier-elimination-threshold-simplified}
        $\card{\homterms{\Aformula_6^c}} = \BigO{\card{\homterms{\aformula}}^2}$.
\end{restatable}


As already stated, this quadratic bound is key in order to obtain an elementary decision procedure.
In particular, this improvement over the ``baseline'' Lemma~\ref{lemma:bound-quantifier-elimination} leads to the following bounds on the elimination of an arbitrary number of threshold counting quantifiers.

\begin{restatable}{lemma}{LemmaBoundQuantTreshMultiple}
    \label{lemma:bound-quantifier-elimination-threshold-d-quant}
    Let $\aformula$ be a formula of Presburger arithmetic with threshold quantifiers.
    There is an equivalent quantifier-free formula~$\Aformula$ such that
    \begin{itemize}
        \setlength{\itemsep}{3pt}
        \item $\card{\linterms{\Aformula}}, \norminf{\linterms{\Aformula}}, \norminf{\homterms{\Aformula}}$ and $\norminf{\fmod{\Aformula}}$ are at most $2^{2^{2^{\BigO{\abs{\aformula}^2}}}}$,
        \item $\card{\homterms{\Aformula}} \leq 2^{2^{\BigO{\abs{\aformula}^2}}}$ and $\card{\fmod{\Aformula}} \leq \abs{\aformula}$.
    \end{itemize}
\end{restatable}

\begin{proof}[Proof idea]
In a nutshell, elementary upper bounds of
Lemma~\ref{lemma:bound-quantifier-elimination-threshold-d-quant}
are obtained by first iterating
Lemma~\ref{lemma:bound-quantifier-elimination-threshold-simplified}
across \emph{all} quantifier elimination rounds.
This results in a doubly exponential bound on the cardinalities
of sets $\homterms{\Aformula_6^c}$ throughout the entire procedure.
With this bound in hand,
exponentiation in the right-hand side of the inequalities of
Lemma~\ref{lemma:bound-quantifier-elimination}
does not blow the parameters above triple exponential.
\end{proof}

%
Theorem~\ref{theorem:threshold} follows by combining
Lemma~\ref{lemma:bound-quantifier-elimination-threshold-d-quant}
with upper bounds on the running time of a single quantifier elimination round.
These upper bounds are all subsumed by the size of the \emph{obtained} formulae,
except possibly for
the procedures of~\Cref{lemma:bound-number-of-orderings,lemma:simplifying-threshold},
and
the model counting procedure of Lemma~\ref{lemma:sharp-sat}.
For~\Cref{lemma:bound-number-of-orderings,lemma:simplifying-threshold}, the running time is only exponential in the size of the \emph{original} formula, and thus runs in polynomial time on the size of the obtained formula, as soon as this formula has size at least exponential.
For~\Cref{lemma:sharp-sat}, observe that the factor~$m$ in the running time
cannot be more than the product of all elements of the set~$\homterms{\Aformula}$.
Hence, the bounds of
Lemma~\ref{lemma:bound-quantifier-elimination-threshold-d-quant} suffice
for a triply exponential time overall.%


\section{Eliminating modulo counting quantifiers}
\label{subsection:elimination-modulo-quantifiers}


Consider a formula~$\Aformula_0 =
\exists^{(\avar,q)}\avarbis\,\aformula$, where $\aformula$ is
quantifier-free.  By definition, $\Aformula_0$ is equivalent to
$\exists \avarter (\avarter \equiv_q \avar \land \exists^{=\avarter}
\avarbis\,\aformula)$, where~$\avarter$ is a variable not occurring
in~$\aformula$. Thus, in order to eliminate a modulo counting
quantifier, it makes sense to piggyback on a quantifier elimination
procedure for the $\exists^{=\avar}\avarbis$ counting quantifier. Due
to space constraints, we only briefly describe the main aspects of
eliminating an $\exists^{=\avar}\avarbis$ counting quantifier, further
details can be found in~\Cref{subsection:exists-equal-elim}. In its
essence, the quantifier elimination procedure mirrors Steps~I-V of the
procedure described in~\Cref{section:quantifier-elimination},
replacing the inequalities in the definitions of $\Psi_3^{i,r}$ and
$\Psi_4^{i,r}$ in Steps~3 and~4 with an equality. However, in addition
in Step~4, instead of setting $\Psi_4^{i,r}$ to $\top$ when there is
an infinite number of solutions, $\Psi_4^{i,r}$ is set to $\bot$ to
capture the semantics of the $\exists^{=\avar}\avarbis$ quantifier.
Apart from that, the quantifier-elimination procedure for the
$\exists^{=\avar}\avarbis$ quantifier inherits all the properties of
the one described in~\Cref{section:quantifier-elimination}.

Let $\Aformula_5^=$ be the formula obtained from performing the quantifier-elimination procedure for the $\exists^{=x}y$ counting quantifier on~$\exists^{= \avarter} \avarbis\,\aformula$, so that $\Aformula_0 \fequiv \exists \avarter (\avarter \equiv_q \avar \land \Aformula_5^=)$.
We have that $\Aformula_5^=$ is defined as 
\begin{center}
$\Aformula_5^= \egdef \bigvee_{i \in [1,o]}\bigvee_{r\colon Z\to [m]} (\Gamma_{i,r} \land \Aformula_5^{i,r})$,
\end{center}
where $Z = \vars{\aformula}$, $m = \lcm(\fmod{\Aformula_1})$ and~$\Gamma_{i,r} = \ordering_i \land (\bigwedge_{\avarfour \in Z} \avarfour \equiv_m r(\avarfour))$ is a conjunction of an ordering~$\ordering_i$ and simple modulo constraints~$\avarfour \equiv_m r(\avarfour)$ with variables from~$Z$ (and so,~$\avarter$-free).
Moreover, $\Aformula_5^{i,r}$ is either $\false$ or a formula of the form 
\begin{equation}
    \label{equation:5ir-form-eq}
    \textstyle m\avarter = \sum_{j=2}^{\ell}( p_j(t_{j}' - t_{j-1}') + r_j ) + m \cdot \sum_{j=1}^{\ell} c_j.
\end{equation}
where the terms $\aterm_1',\dots,\aterm_{\ell}'$ are from~$T \cup \{0\}$ (where $T$ is defined as in Step II of~\Cref{section:quantifier-elimination}), and 
hence~$\avarter$-free.
Analogously to~\Cref{section:elimination-treshold-quantifiers}, \Cref{claim:zed-almost-free}, the following property holds. 

\begin{claim}
    \label{claim:zed-almost-free-bis}
    In~$\Aformula_5^{=}$,~$\avarter$ only appears on the left hand side of 
    equalities of the form~\eqref{equation:5ir-form-eq}. 
\end{claim}
We manipulate~$\exists \avarter (\avarter \equiv_q \avar \land \Aformula_5^=)$ with the following steps, denoted by VI and VII to stress the fact that they are performed after the five steps of the $\exists^{=x}y$ quantifier elimination procedure.%

\vspace{-2ex}
\subparagraph*{Step VI: Subdivide the formula according the residue classes (again).}
To efficiently eliminate the existential quantifier of the formula~$\exists \avarter (\avarter \equiv_q \avar \land \Aformula_5^=)$, we first guess the residue classes of all variables in $Z \cup \{\avar\}$ modulo~$mq$ (instead of just~$m$, as done in~$\Aformula_5^=$ for the variables in~$Z$).
    \begin{itemize}
        \setlength{\itemsep}{2pt}
        \item Let $\aformulater = \bigvee_{s \colon (Z \cup \{\avar\}) \to [mq]} \bigvee_{i \in [1,o]}\bigvee_{r \colon Z \to [m]} ((\bigwedge_{\avarfour \in Z \cup \{\avar\}} \avarfour \equiv_{mq} s(\avarfour)) \land \Gamma_{i,r} \land \Aformula_5^{i,r})$.
        \item For every $s \colon (Z\cup\{\avar\}) \to [mq]$, every $i \in [1,o]$ and every $r \colon Z \to [m]$, consider the disjunct 
        $(\bigwedge_{\avarfour \in Z \cup \{\avar\}} \avarfour \equiv_{mq} s(\avarfour)) \land \Gamma_{i,r} \land \Aformula_5^{i,r}$ of the formula~$\aformulater$
        and evaluate every modulo constraint $\avarfour \equiv_m r(\avarfour)$ in $\Gamma_{i,r}$ ($\avarfour \in Z$) to $\true$ or $\false$, according to the truth of $s(\avarfour) \equiv_m r(\avarfour)$.
    \end{itemize}
    Since every function $r \colon Z \to [m]$ can be seen as a partial function from~$Z \cup \{\avar\}$ to~$[mq]$, after the two steps above, for every~$s \colon (Z\cup\{\avar\}) \to [mq]$ and~$i \in [1,o]$, all but one disjunct of the subformula $\bigvee_{r \colon Z \to [m]} ((\bigwedge_{\avarfour \in Z \cup \{\avar\}} \avarfour \equiv_{mq} s(\avarfour)) \land \Gamma_{i,r} \land \Aformula_5^{i,r})$ of~$\aformulater$ evaluate $\false$.
    We conclude that $\aformulater$ is equivalent to 
    \begin{center}
        $\bigvee_{i \in [1,o]}\bigvee_{s\colon (Z\cup \{\avar\}) \to [mq]} (\Gamma_{i,s} \land \Aformula_5^{i,s})$,
    \end{center}
    where $\Gamma_{i,s} \egdef \ordering_i \land (\bigwedge_{\avarfour \in Z\cup\{\avar\}} \avarfour \equiv_{mq} s(\avarfour))$ and $\Aformula_5^{i,s}$ 
    is~$\false$ or a formula as described~in~\eqref{equation:5ir-form-eq}.

    Let~$\Aformula_6^{(\avar,q)} \egdef \bigvee_{i \in [1,o]}\bigvee_{s\colon (Z\cup \{\avar\}) \to [mq]} \exists \avarter (\avarter \equiv_q \avar \land \Gamma_{i,s} \land \Aformula_5^{i,s})$, the following holds:

    \begin{restatable}{claim}{ClaimModAFiASix}
        \label{claim:mod-a5-a6}
        $\exists \avarter (\avarter \equiv_q \avar \land \Aformula_5^=)$ $\fequiv$ $\Aformula_6^{(\avar,q)}$.
    \end{restatable}
    
\vspace{-2ex}
\subparagraph*{Step VII: Eliminate existential quantifiers.}
We conclude the procedure by manipulating each disjunct~of~$\Aformula_6^{(x,q)}$ separately. 
Fix~$i \in [1,o]$ and $s \colon (Z \cup \{\avar\}) \to [mq]$.
We aim at defining a quantifier-free formula~$\aformulater_{i,s}$ equivalent to the disjunct $\exists \avarter (\avarter \equiv_q \avar \land \Gamma_{i,s} \land \Aformula_5^{i,s})$ of~$\Aformula_6^{(x,q)}$.

\begin{itemize}
    \setlength{\itemsep}{2pt}
    \item  If $\Aformula_5^{i,s} = \false$ then let $\aformulater_{i,s} = \false$. 
    \item  Else, $\Aformula_5^{i,s}$ has the form in~\Cref{equation:5ir-form-eq}.
    By modular arithmetic, $\avarter \equiv_q \avar \fequiv m\avarter \equiv_{mq} m\avar$. Consider the formula 
        $m\avar \equiv_{mq} \sum_{j=2}^{\ell}( p_j(t_{j}' - t_{j-1}') + r_j ) + m \cdot \sum_{j=1}^{\ell} c_j$.
    By~\Cref{claim:zed-almost-free-bis}, this formula has variables from~$Z \cup \{\avar\}$.
    Evaluate this formula on~$s$. 
    If it is found to be equivalent to~$\false$, let $\aformulater_{i,s} \egdef\,\false$. Otherwise, let $\aformulater_{i,s} \egdef \Gamma_{i,s}$.
    \item Let $\Aformula_7^{(\avar,q)} \egdef \bigvee_{i \in [1,o]}\bigvee_{s\colon (Z\cup \{\avar\}) \to [mq]} \aformulater_{i,s}$.
\end{itemize}
After defining~$\Aformula_7^{(\avar,q)}$, the procedure ends.
Notice that all the disjuncts of~$\Aformula_7^{(\avar,q)}$ are either~$\false$ or $\Gamma_{i,s} = \ordering_i \land \bigwedge_{\avarfour \in Z \cup \{\avar\}} \avarfour \equiv_{mq} s(\avarfour)$, where $\vars{\ordering_i} \subseteq Z \cup \{\avar\}$. The (in)equalities appearing in~$\ordering_i$ are of the form $\aterm \lhd \aterm'$, where $\lhd \in \{<,=\}$ and~$\aterm,\aterm' \in T \cup \{0\}$. Exactly as in the case of threshold quantifiers, this leads to~$\card{\homterms{\Aformula_7^{(\avar,q)}}} \leq (\card{T} \cup \{0\})^2$, which 
ultimately leads to a~\threeexptime running time for the QE procedure.

\begin{restatable}{claim}{ClaimModASixASep}
    \label{claim-mod-A6-A7}
    $\Aformula_6^{(\avar,q)} \fequiv \Aformula_7^{(\avar,q)}$.
    The formula $\Aformula_7^{(\avar,q)}$ is quantifier-free.
\end{restatable}

\subsection*{Proof idea for \Cref{theorem:modulo}}
The key role in the analysis is played by the following lemma.

\begin{restatable}{lemma}{LemmaBoundQuantifierEliminationModuloSimplified}
    \label{lemma:bound-quantifier-elimination-modulo-simplified}
        $\card{\homterms{\Aformula_7^{(\avar,q)}}} = \BigO{\card{\homterms{\aformula}}^2}$.
\end{restatable}

When eliminating an arbitrary number of~$\exists^{(\avar,q)}\avarbis$,~\Cref{lemma:bound-quantifier-elimination-modulo-simplified} leads to the following result.

\begin{restatable}{lemma}{LemmaBoundQuantifierElimDQUANT}
    \label{lemma:bound-quantifier-elimination-modulo-d-quant}
    Let $\aformula$ be a formula of Presburger arithmetic with threshold quantifiers.
    There is an equivalent quantifier-free formula~$\Aformula$ such that
    \begin{itemize}
        \setlength{\itemsep}{3pt}
        \item $\card{\linterms{\Aformula}}, \norminf{\linterms{\Aformula}}, \norminf{\homterms{\Aformula}}$ and $\norminf{\fmod{\Aformula}}$ are at most $2^{2^{2^{\BigO{\abs{\aformula}^2}}}}$,
        \item $\card{\homterms{\Aformula}} \leq 2^{2^{\BigO{\abs{\aformula}^2}}}$ and $\card{\fmod{\Aformula}} \leq \abs{\aformula}$.
    \end{itemize}
\end{restatable}

\Cref{lemma:bound-quantifier-elimination-modulo-d-quant} allows us to establish that the QE procedure for Presburger arithmetic with modulo counting quantifiers runs in~\threeexptime. The proof follows the pattern of~\Cref{theorem:threshold}.



\section{Conclusion}
\label{section:conclusion}

We developed a QE procedure for Presburger
arithmetic extended with the unary threshold counting quantifier
$\exists^{\ge c}y$ that runs in \threeexptime, i.e., at no additional
cost compared to standard QE procedures for
Presburger arithmetic, see e.g.~\cite{Opp78}. From the estimation of
the growth of the constants occurring in our QE
procedure, using standard relativisation arguments, see
e.g.~\cite{Weispfenning90}, we can derive that the decision problem
for Presburger extended with the $\exists^{\ge c}$ quantifier is
in \twoexpspace. This matches the complexity of deciding standard
Presburger arithmetic closely. Indeed, the latter is complete for the
complexity class $\textsc{STA}(*,2^{2^{\poly n}},O(n))$~\cite{Ber80}.
Fully settling the complexity of Presburger arithmetic extended with
$\exists^{\ge c}y$ will likely require generalising the $\textsc{STA}$
complexity measure, which we leave as an interesting avenue for
further investigation.

Our QE procedure is based on a QE procedure for the more general $\exists^{\ge x}y$ counting
quantifier that we developed in this paper. While the latter procedure
slightly improves the QE procedure given by
Schweikardt~\cite{Schweikardt05}, it still only runs in non-elementary
time. We have pinpointed precisely at where the non-elementary growth
occurs. It remains to be seen whether our QE
procedure can be further improved, or whether, possibly based on the
insights obtained from our QE procedure, a
non-elementary lower bound for Presburger arithmetic extended with the
$\exists^{\ge x}y$ quantifier can be established.

\bibliography{bibliography}

\appendix
\section{Missing proofs from \Cref{section:quantifier-elimination}}

\LemmaMakeModuloSimple*
\begin{proof}
    Let $Z = \vars{\aterm \equiv_q 0}$. We guess the residue classes of the variables in~$Z$, as shown in the right hand side of the following equivalence:
    \begin{center}
        $\aterm \equiv_q 0 \fequiv \bigvee_{r \colon Z \to [q]} (\aterm \equiv_q 0 \land \bigwedge_{\avarter \in Z} \avarter \equiv_q r(\avarter))$.
    \end{center}
    Fix $r \colon Z \to [q]$, and consider the disjunct $(\aterm \equiv_q 0 \land \bigwedge_{\avarter \in Z} \avarter \equiv_q r(\avarter))$. Let~$\avar$ be a variable occurring in~$\aterm$.
    As~$r$ assigns to~$\avar$ a residue class modulo~$q$, the following equivalence holds:
    \begin{center}
        $\aterm \equiv_q 0 \land \bigwedge_{\avarter \in Z} \avarter \equiv_q r(\avarter)$
        $\fequiv$
        $(\aterm \equiv_q 0)\substitute{\avar}{r(\avar)} \land \bigwedge_{\avarter \in Z} \avarter \equiv_q r(\avarter)$.
    \end{center}
    Therefore, by substituting in~$\aterm$ every variable~$\avar$ with~$r(\avar)$, we derive
    \begin{center}
        $\aterm \equiv_q 0 \land \bigwedge_{\avarter \in Z} \avarter \equiv_q r(\avarter)$
        $\fequiv$
        $r(\aterm) \equiv_q 0 \land \bigwedge_{\avarter \in Z} \avarter \equiv_q r(\avarter)$.
    \end{center}
    Since~$r(\aterm) \equiv_q 0$ does not have free variables (i.e.~it is a statement), it is equivalent to~$\true$ or~$\false$. 
    Let $\aformulabis_r \in \{\true,\false\}$ such that $r(\aterm) \equiv_q 0 \fequiv \aformulabis_r$, and let 
    $\aformulabis = \bigvee_{r \colon Z \to [q]} (\aformulabis_r \land \bigwedge_{\avarter \in Z} \avarter \equiv_q r(\avarter))$. 
    The formula~$\aformulabis$ satisfies the required properties.
\end{proof}

\claimPsiOne*
\begin{claimproof}
    We show the following sequence of equivalences:
    \begin{align}
        \Aformula_0 &\ \fequiv\ \exists^{\geq \avar}\avarbis\, \aformula'\label{psi1:e1}\\
                    &\ \fequiv\ \exists^{\geq \avar}\avarbis\, \exists \avarter(\avarter = k\avarbis \land \aformula'\substitute{k\avarbis}{\avarter}) &\text{for some fresh variable}~\avarter\label{psi1:e2}\\
                    &\ \fequiv\ \exists^{\geq \avar}\avarter\, \exists \avarbis\,(\avarter = k\avarbis \land \aformula'\substitute{k\avarbis}{\avarter})\label{psi1:e3}\\
                    &\ \fequiv\ \exists^{\geq \avar}\avarter\, ( \avarter \equiv_k 0 \land \aformula'\substitute{k\avarbis}{\avarter})\label{psi1:e4}\\
                    &\ \fequiv\ \exists^{\geq \avar}\avarbis\, ( \avarbis \equiv_k 0 \land \aformula'\substitute{k\avarbis}{\avarbis})= \Aformula_1. \label{psi1:e5}
    \end{align}
    The equivalence~\eqref{psi1:e1} holds, because the rewrite rules used to produce $\aformula'$ from $\aformula$ come from biconditional axioms of (modular) arithmetic, e.g.~$a \equiv_b c \fequiv ka \equiv_{kb} kc$, for all $k \geq 1$.
    In~$\aformula'$, all non-zero coefficients of~$y$ are either~$k$ or~$-k$.
    This directly establishes equivalence~\eqref{psi1:e2}.
    In the formula~$\aformula'\substitute{k\avarbis}{\avarter}$ (and thus in~$\Aformula_1$) all non-zero coefficients of~$y$ are either~$1$ or~$-1$.
    The equivalence~\eqref{psi1:e3}  holds as the expression $\avarter = k \avarbis$ induces a bijection between all possible values of~$\avarbis$ and~$\avarter$. 
    The equivalence~\eqref{psi1:e4} holds as~$\avarbis$ does not occur in $\aformula'\substitute{k\avarbis}{\avarter}$. Notice that $z \equiv_k 0 \fequiv \exists y\, z = ky$.
    The equivalence~\eqref{psi1:e5} follows as we rename $z$ by~$y$.
\end{claimproof}

\lemmaBoundNumberOfOrderings*
\begin{proof}
We first show the existence of a family of orderings
with required properties. This part of the proof relies
on the insight 
that $n$~hyperplanes split~$\mathbb R^d$ into $\mathcal O(n^d)$ regions.
This is the basis of multiple ``geometric'' decision procedures and algorithms;
see, e.g.,~\cite{Sca84,Woods15}.

\begin{claim}
\label{claim:n-power-d-for-signs}
    Given $s_1, \dots, s_m$ linear terms over~$d$ variables, 
    there are at most
    $\mathcal O(m^d)$ conjunctions of the form
    \begin{equation*}
    (s_1 R_1 0) \land
    (s_2 R_2 0) \land \ldots \land
    (s_m R_m 0),
    \end{equation*}
    where $R_i \in \{{<}, {=}, {>}\}$,
    that are satisfiable over the reals~$\mathbb R$.
\end{claim}
\begin{claimproof}
This is a small variation of the classic proof, using double induction.
For $d = 1$, the number of such conjunctions is clearly
at most $2 m + 1$, because $m$ points can split the line
into (at most) $m + 1$ finite or infinite open intervals
and $m$ points themselves.

For larger $d$, we proceed as follows.
We assume an $\mathcal O(m^{d-1})$ bound for dimension~$d-1$.
For $m = 1$, the number of satisfiable conjunctions is at most~$3$.
Let us deal with larger $m$ now.
Suppose we have already computed (or, rather, bounded from above)
the number of satisfiable conjunctions of terms $s_1, \ldots, s_{j-1}$,
and suppose this number is $N$.
Consider what happens when the term $s_j$ is added to them.
Each region of $\mathbb R$ that corresponds to one of the~$N$~satisfiable
conjunctions of $s_1, \ldots, s_{j-1}$ can be ``cut''
by the new term into at most $3$~regions,
according to whether $s_j < 0$, $s_j = 0$, or $s_j > 0$
(and thus adding $2$~new regions).
This is the only way new regions, and thus satisfiable conjunctions
of the form
$
    (s_1 R_1 0) \land
    \ldots \land
    (s_j R_j 0)
$ are composed. However, we can now observe that the number of regions
that are ``cut'' is not, in general, as big as~$N$.
Indeed, the number of regions that are ``cut'' is bounded from
above by the number of regions inside the set $\{ \vec x \in \mathbb R^d
\mid s_j = 0 \}$ formed by the terms $s_1, \ldots, s_{j-1}$.
But this number is $\mathcal O((j-1)^{d-1})$ by the bound
for dimension~$d-1$.
Therefore, for dimension~$d$ we obtain an overall bound of
\begin{equation*}
\mathcal O(1) + \sum_{j=2}^{m} 2 \cdot \mathcal O( (j-1)^{d-1} ) = \mathcal O(m^d).
\qedhere
\end{equation*}
\end{claimproof}

\begin{claim}
\label{claim:n-power-d-for-ineq}
    Given $\aterm_1, \dots, \aterm_n$ linear terms over~$d$ variables, 
    there are at most
    $\mathcal O(n^{2 d})$ orderings
    that satisfy properties~(I) and~(II).
\end{claim}

\begin{claimproof}
This is a consequence of Claim~\ref{claim:n-power-d-for-signs}.
Indeed, we can form $m = \binom{n}{2}$ terms of the form $t_i - t_j$, $i \ne j$.
For any valuation to the $d$~variables,
the signs of these $m$ terms determine an ordering of $t_1, \ldots, t_n$,
satisfiable over the reals.
Therefore, we obtain $\mathcal O(m^d) = \mathcal O(n^{2 d})$ orderings in total.
\end{claimproof}

Given Claim~\ref{claim:n-power-d-for-ineq},
let us now proceed to the second, algorithmic part of the proof.
The idea can be seen as dynamic programming.

Our algorithm runs as follows.
Let $t_1, \ldots, t_n$ be the terms from the statement of the lemma.
We construct several families of orderings, incrementally:
family $\mathcal F_j$ is the required family
for the the first $j$~terms, $t_1, \ldots, t_j$.

For $j = 1$, the family $F_1$ is trivial.
For $j = 2, \ldots, n$, we compute $F_j$ from $F_{j-1}$ as follows.
Start from $F_j = \varnothing$. For each ordering $o$ from $F_{j-1}$, enumerate all
possible positions to insert $t_j$ into it. There are at most $2 j - 1$ possible
options here; their precise number depends on the number of equalities
(as opposed to inequalities)
among the relations on $t_1, \ldots, t_j$. For
each of these options, check if the resulting ordering is satisfiable
when the variables are interpreted over $\mathbb R$, using any polynomial-time
algorithm for linear programming.
If so, add it to $F_j$, otherwise just skip it.
In the end, $F_n$ is a family of orderings with the required properties.

Let us analyse this algorithm.
From the (non-algorithmic) part of the lemma, already proved above
as Claim~\ref{claim:n-power-d-for-ineq},
we know that there are
$\mathcal O(j^{2d})$ orderings for $j$~terms.
For each of these orderings, we will try
to insert $t_{j+1}$ at $\mathcal O(j)$ possible positions.
So there are $\mathcal O(j^{2d+1})$
satisfiability checks to run.
Over all~$j$, this is $\mathcal O(n^{2d+2})$ checks.

Notice that it is \emph{sufficient} for us to look at satisfiability over the
reals or rationals here,
as long as the number of orderings we get is not too high.
Indeed, if some ordering is satisfiable over $\mathbb R$
(or, equivalently, over $\mathbb Q$) but not satisfiable over $\mathbb Z$,
then we may still include it.
This means that we can rely on polynomial-time algorithms for linear
programming. Each instance has at most $n$ constraints over $d$ variables.
The bit size of each coefficient be bounded by $b \egdef \log \max_j \norminf{t_j}$.
Therefore, a satisfiability check for a system of constraints of this
form can be run in time $\poly{b, n, d}$. The overall running time for the
entire algorithm is
\begin{equation*}
\mathcal O(n^{2d+2}) (b n d)^{\mathcal O(1)} =
n^{\mathcal O(d)} (n d)^{\mathcal O(1)} b^{\mathcal O(1)} =
n^{\mathcal O(d)} \log \max_j \norminf{t_j}^{\mathcal O(1)}.
\qedhere
\end{equation*}
\end{proof}

\claimPsiTwo*
\begin{claimproof}
    Let~$\aformulabis_{\text{rem}}$ be the formula $\bigvee_{r\colon Z\to [m]} \bigwedge_{z \in Z} z \equiv_m r(z)$ whose disjuncts represent a combination of residue classes modulo~$m$ for the variables in~$Z$. 
    We have 
    \begin{align}
        \Aformula_1 &\ \fequiv\ \aformulabis_{\text{ord}} \land \aformulabis_{\text{rem}} \land \Aformula_1 \label{psi2:e1}\\
        &\ \fequiv\ \Aformula_2. \label{psi2:e2}
    \end{align}
    The equivalence~\eqref{psi2:e1} follows from the fact that both the formulae~$\aformulabis_{\text{ord}}$ and $\aformulabis_{\text{rem}}$ are tautologies.
    From~$\aformulabis_{\text{ord}} \land \aformulabis_{\text{rem}} \land \Aformula_1$,
    we distribute the conjunctions over the disjunctions given by~$\bigvee_{i \in [1,o]}$ and~$\bigvee_{r\colon Z\to [m]}$, which shows equivalence~\eqref{psi2:e2}.
\end{claimproof}

\claimPsiThreeSimpleMod*
\begin{claimproof}
Let $\aformulafour$ be fixed.
We recall that the formulae $O_i$ were constructed based on the set of terms
that includes~$0$. This means, in particular, that for all assignments
that satisfy the conjunction $\Gamma_{i,r} \land \aformulafour$
(if any exist)
the truth value of all inequalities that occur in the formula~$\aformulabis$
is the same. Indeed, for inequalities not involving the variable~$y$ this
is because the formula~$O_i$ asserts or implies the sign of every linear term.
For inequalities involving~$y$, this is due to our choice of the set~$\segments(\avarbis, \ordering_i)$.

We now consider modulo constraints that occur in~$\aformulabis$.
Those of them where the variable~$y$ does not appear also evaluate to
just true or false on all assignments satisfying $\Gamma_{i,r} \land \aformulafour$,
because $r$~specifies residue classes modulo~$m = \lcm(\fmod{\aformulabis})$ for all
variables except~$y$. Since $y$~can only occur with coefficient $1$ or $-1$,
all the remaining modulo constraints become simple, i.e., take the form
$y \equiv_q r$ for some $q \in \fmod{\aformulabis}$.

To sum up, replacing all constraints in the $\aformulabis$ part of the formula
$\Gamma_{i,r} \land \aformulafour \land \aformulabis$ with their truth values
or their simplified form, as described above, we obtain an equivalent formula
$\Gamma_{i,r} \land \aformulafour \land \aformulabis_{\aformulafour}^{i,r}$,
as required.
\end{claimproof}

\claimPsiThree*
\begin{claimproof}
    Let~$i \in [1,o]$, $r \colon Z \to [m]$. 
    Establishing~$\Gamma_{i,r} \land \Aformula_1$ $\fequiv$ $\Gamma_{i,r} \land \Aformula_{3}^{i,r}$ suffices,
    where~$\Aformula_1 = \exists^{\geq \avar}\avarbis\,\aformulabis$ and~$\Aformula_3^{i,r} = \exists \avar_0\dots\exists\avar_{2\ell} \left(
        \avar \leq \avar_0 + \dots + \avar_{2\ell} \land \bigwedge_{j \in [0,2\ell]} \exists^{\geq \avar_j} \avarbis (\aformulafour_j \land \aformulabis^{i,r}_{\aformulafour_j})\right)$.
    Directly from~\Cref{claim:psi3-simpl-mod}, 
    The formula $\Gamma_{i,r} \land \Aformula_{3}^{i,r}$ is equivalent to 
    \begin{center}
    $\aformulater \egdef \Gamma_{i,r} \land \exists \avar_0\dots\exists\avar_{2\ell} \left(
        \avar \leq \avar_0 + \dots + \avar_{2\ell} \land \bigwedge_{j \in [0,2\ell]} \exists^{\geq \avar_j} \avarbis (\aformulafour_j \land \aformulabis)\right)$
    \end{center}
    where we notice that all the formulae of the form~$\aformulabis^{i,r}_{\aformulafour_j}$ are substituted with~$\aformulabis$.
    Proving the equivalence~$\Gamma_{i,r} \land \Aformula_1 \fequiv \aformulater$ is rather straightforward.
    \ProofLeftarrow Let $\aeval$ be an assignment such that $\aeval \models \Aformula_1 \fequiv \aformulater$.
    Therefore, there are values $v_0,\dots,v_{2\ell}$ for the variables $\avar_0,\dots,\avar_{2\ell}$ such that 
    $\aeval[v_0/\avar_0,\dots,v_{2\ell}/\avar_{2\ell}] \models \avar \leq \avar_0 + \dots + \avar_{2\ell} \land \bigwedge_{j \in [0,2\ell]} \exists^{\geq \avar_j} \avarbis (\aformulafour_j \land \aformulabis)$. Since the variables $\avar_0,\dots,\avar_{2\ell}$ do not appear in $\aformulafour_j \land \aformulabis$, we have
        $\aeval \models \avar \leq v_0 + \dots + v_{2\ell} \land \bigwedge_{j \in [0,2\ell]} \exists^{\geq v_j} \avarbis (\aformulafour_j \land \aformulabis)$.
    Lastly, in view of the definition of the set~$\segments(\avarbis,\ordering_i)$, 
    given~$\aformulafour,\aformulafour' \in \segments(\avarbis,\ordering_i)$ there is no value~$v$ for $\avarbis$ such that $\aeval\substitute{\avarbis}{v} \models \aformulafour \land \aformulabis$ and $\aeval\substitute{\avarbis}{v} \models \aformulafour' \land \aformulabis$.
    We conclude that there are at least $\sum_{j = 0}^{2\ell} v_j$ distinct values~$v$ for $y$ such that 
    $\aeval\substitute{\avarbis}{v} \models \aformulabis$, and thus $\aeval \models \Gamma_{i,r} \land \Aformula_1$, directly from $\aeval \models \avar \leq v_0 + \dots + v_{2\ell}$.

    \ProofRightarrow 
    Suppose $\aeval \models \Gamma_{i,r} \land \Aformula_1$. So, there are at least $\aeval(\avar)$ distinct values~$v$ for $\avarbis$ such that $\aeval\substitute{\avarbis}{v} \models \aformulabis$.
    Given $j \in [0,2\ell]$, let $v_j$ be the number distinct values for $\avarbis$ such that 
    $\aeval\substitute{\avarbis}{v} \models \aformulafour_j \land \aformulabis$.
    Again from the fact that, given~$\aformulafour,\aformulafour' \in \segments(\avarbis,\ordering_i)$ there is no value~$v$ for~$\avarbis$ such that $\aeval\substitute{\avarbis}{v} \models \aformulafour \land \aformulabis$ and $\aeval\substitute{\avarbis}{v} \models \aformulafour' \land \aformulabis$, 
    we conclude that $\aeval(\avar) \leq \sum_{j = 0}^{2\ell} v_j$.
    Thus, $\aeval[v_0/\avar_0,\dots,v_{2\ell}/\avar_{2\ell}] \models \avar \leq \avar_0 + \dots + \avar_{2\ell} \land \bigwedge_{j \in [0,2\ell]} \exists^{\geq \avar_j} \avarbis (\aformulafour_j \land \aformulabis)$, and so $\aeval \models \aformulater$.
\end{claimproof}

\lemmaSharpSat*
\begin{proof}
We notice that each of
the numbers~$p_j$ and $r_j'$ are defined by counting the number of solutions of~$\aformulabis_{\aformulafour}^{i,r}$ in a finite interval. This formula is quantifier-free, and~$\vars{\aformulabis_{\aformulafour}^{i,r}} = \{\avarbis\}$.
If the formula and the interval bounds are given as input, then
counting these numbers is a~\sharpP problem.
One can enumerate all possible values of~$y$ and check each of them
against the formula.
Since $\overline{u}_j - \underline{u}_j \le m$, there are at most
$m$~values to check in each of the intervals.
\end{proof}

\claimPsiFourInfSol*
\begin{claimproof}
We focus without loss of generality on the case
$\aformulafour \egdef (\avarbis < t_1')$.
Suppose
the formula $\exists \avarbis\, (\aformulafour \land \aformulabis_{\aformulafour}^{i,r})$
is satisfiable.
It suffices to show that,
for \emph{every} assignment $\mu$ that satisfies the formula~$\Gamma_{i,r}$,
the formula $\Aformula^{i,r}_3$ is satisfied by $\mu$ too.

Recall that the formula $\Aformula^{i,r}_3$ was
defined previously as
\begin{equation*}
\exists \avar_0\dots\exists\avar_{2\ell} \left(
            \avar \leq \avar_0 + \dots + \avar_{2\ell} \land \bigwedge_{j \in [0,2\ell]} \exists^{\geq \avar_j} \avarbis (\aformulafour_j \land \aformulabis^{i,r}_{\aformulafour_j})\right).
\end{equation*}
Let us take an arbitrary assignment $\mu$ that satisfies the formula~$\Gamma_{i,r}$.
Our goal is to show that the formula above is satisfied by~$\mu$ too.
We have our
$\aformulafour \in \segments(y,\ordering_i)$,
and we assume without loss of generality that $\aformulafour = \aformulafour_0$.
Consider a new assignment~$\mu'$ obtained from~$\mu$ by updating
the values of all $x_1, \ldots, x_{2 \ell}$ to~$0$
and the value of $x_0$ to~$x$.
Clearly, $\mu'$ satisfies the inequality
$\avar \leq \avar_0 + \dots + \avar_{2\ell}$.
For each $j \in [1, 2\ell]$, the formula
$\exists^{\geq \avar_j} \avarbis (\aformulafour_j \land \aformulabis^{i,r}_{\aformulafour_j})$
is now satisfied too, because $\mu'(\avar_j) = 0$,
so
it remains to argue that $\mu'$ satisfies
$\exists^{\geq \avar} \avarbis (\aformulafour_0 \land \aformulabis^{i,r}_{\aformulafour_0})$
where $\aformulafour_0 = \aformulafour$.

Let us capitalise on the fact that $\aformulafour$ is
just a single inequality, $\avarbis < t_1'$.
First, note that setting the value of~$y$ to any integer strictly smaller
than $\mu'(t_1')$ satisfies $\aformulafour$.
Second, recall from Step~III (and~\Cref{claim:psi3-simpl-mod})
that the formula $\aformulabis_{\aformulafour}^{i,r}$
is 
a Boolean combination of simple modulo constraints
with $\vars{\aformulabis_{\aformulafour}^{i,r}} = \{\avarbis\}$
and that
$\fmod{\aformulabis_{\aformulafour}^{i,r}} \subseteq \fmod{\aformulabis}$.
So our previous choice of $m = \lcm(\fmod{\Aformula_1})$
is a multiple of $\lcm(\fmod{\aformulabis_{\aformulafour}^{i,r}})$,
and thus the set of all assignments (for $y$) that satisfy
the formula~$\aformulabis_{\aformulafour}^{i,r}$ is periodic
with period~$m$. Importantly, this set is non-empty
because
the formula $\exists \avarbis\, (\aformulafour \land \aformulabis_{\aformulafour}^{i,r})$
was assumed to be satisfiable (and thus
$\aformulabis_{\aformulafour}^{i,r} \not\fequiv \false$).
Therefore, this set contains all elements of some infinitely descending sequence
with difference~$m$.
Therefore, setting the value of~$y$ to \emph{any} such element
smaller than $\mu'(t_1')$ satisfies the formula
$(\aformulafour_0 \land \aformulabis^{i,r}_{\aformulafour_0})$.
This completes the proof, as there are indeed infinitely many such values,
whilst
the formula
$\exists^{\geq \avar} \avarbis (\aformulafour_0 \land \aformulabis^{i,r}_{\aformulafour_0})$
asserts the existence of just $\mu'(x)$ of them, for some $\mu'(x) \in \Zed$.
\end{claimproof}

\claimPsiFourRemoveCqEq*
\begin{claimproof}
For every assignment~$\aeval$ to variables other than~$y$,
there cannot be more than one value of~$y$ that
satisfies the formula
$\aformulafour \land \aformulabis_{\aformulafour}^{i,r}$.
Indeed, since $\aformulafour$ is $y = \aterm_j'$,
then the only possible choice for $y$ is $\aeval(\aterm_j')$.
If this choice satisfies $\aformulabis_{\aformulafour}^{i,r})$,
then $z$ can be chosen to be any integer less than or equal to~$1$.
Otherwise, there is no update to $\aeval$ by any value assigned to~$y$
that would make the formula
$\aformulafour \land \aformulabis_{\aformulafour}^{i,r}$
satisfied,
in which case
$z$ can be set to all non-positive integers (and only to them).
\end{claimproof}

\claimPsiFourRemoveCqIneq*
\begin{claimproof}[Proof of~\Cref{claim:psi4-remove-cq-ineq}]
For every assignment~$\aeval$ to variables other than~$y$,
take $L = \aeval(t_{j-1}') + 1$ and $U = \aeval(t_{j}')$
and consider the segment of integers $[L, U-1]$,
using the convention $[a, b] = \varnothing$ if $a > b$.
These are all the potential values of~$y$ that
satisfy the formula
$\aformulafour$,
which under the conditions of the Claim has the form
$\aterm_{j-1}' < \avarbis \land \avarbis < \aterm_{j}'$.
We need to determine how many of these values actually satisfy
the larger formula~%
$\aformulafour \land \aformulabis_{\aformulafour}^{i,r}$,
and this number will be the maximum possible value
attained by the variable~$z$.

As in the proof of~\Cref{claim:psi4-inf-sol},
recall from Step~III (and~\Cref{claim:psi3-simpl-mod})
that the formula $\aformulabis_{\aformulafour}^{i,r}$
is 
a Boolean combination of simple modulo constraints
with $\vars{\aformulabis_{\aformulafour}^{i,r}} = \{\avarbis\}$
and that
$\fmod{\aformulabis_{\aformulafour}^{i,r}} \subseteq \fmod{\aformulabis}$.
So our previous choice of $m = \lcm(\fmod{\Aformula_1})$
is a multiple of $\lcm(\fmod{\aformulabis_{\aformulafour}^{i,r}})$,
and thus the set of all assignments (for $y$) that satisfy
the formula~$\aformulabis_{\aformulafour}^{i,r}$ is periodic
with period~$m$. Formally, denote
$S = \{ n \in \Zed \mid (y \mapsto n) \models \aformulabis_{\aformulafour}^{i,r} \}$
and note that $n \in S$ iff $n + m \in S$.

The technical hurdle we need to overcome in this proof
is that the values of $L$ and $U$ depend on the assignment~$\aeval$.
Importantly, it suffices to consider assignments that satisfy~$\Gamma_{i,r}$,
i.e., the formula
$\ordering_i \land (\bigwedge_{z \in Z} z \equiv_m r(z))$.
We focus on the modulo constraints in this formula and, from now on,
we assume that~$\aeval$ satisfies all of them.
Our goal is to compute the cardinality of the set
$[L, U-1] \cap S$, which by the arguments above is exactly
the number of variable assignments to~$y$ that satisfy
$\aformulafour \land \aformulabis_{\aformulafour}^{i,r}$,
if all other variables have already been assigned values by~$\aeval$.

Let $L' = r(t_{j-1}') + 1$ and
$U' = \min \{ r(t_j') + m \cdot h \mid h \in \Zed \} \cap [L', +\infty)$.
These two numbers are almost the same as $\underline{u}_j$ and $\overline{u}_j$,
respectively, but we will use the capital letter notation to keep
symbols for different segment endpoints uniform.
Now $L' \le U'$, $\card{[L', U'-1]} = U' - L' \in [0, m-1]$ and,
because of our assumption about~$\aeval$,
$L \equiv L' \bmod m$ and $U \equiv U' \bmod m$.
(For the proof of these congruences, observe that, firstly,
 $r(z) \equiv \aeval(z) \bmod m$ because $\aeval$~satisfies $\Gamma_{i,r}$.
 This implies that $a r(z) \equiv a \aeval(z) \bmod m$ for all $a \in \Zed$.
 Summing up several congruences of this kind results in another congruence,
 of the form $r(t) \equiv \aeval(t) \bmod m$. Setting
 $t = t_{j-1}' + 1$ and $t = t_j'$ concludes the proof.)

We are now ready to compute the cardinality of $[L, U-1] \cap S$.
As $S$ is periodic with period~$m$, we will split $[L, U-1]$ into
two disjoint parts: $[L, U-1] = [L, L^* - 1] \cup [L^*, U-1]$,
where $L^*$ is the largest integer not exceeding $U$ and congruent to~$L$
modulo~$m$. We consider each part separately:
\begin{itemize}
\item
As $L^*$ is congruent to $L$ modulo~$m$, it is clear that the integer
segment $[L, L^*-1]$ consists of zero, one, two or more copies of
a full period of~$S$. Therefore,
\begin{equation*}
\card{([L, L^*-1] \cap S)}
=
\card{([0,m-1] \cap S)} \cdot \frac{L^* - L}{m}
=
p_j \cdot \frac{L^* - L}{m}
.
\end{equation*}
\item
For the second part, observe that all three numbers $L$, $L'$, and $L^*$
are congruent modulo~$m$; similarly, $U'$ and $U$ are congruent modulo~$m$.
By definition of $L^*$, we have $\card{[L^*, U-1]} = U - L^* \in [0; m-1]$.
Therefore, the following two constraints hold:
\begin{gather*}
\card{[L^*, U - 1]} = \card{[L', U' - 1]} \quad\text{and}\\
L^* \equiv L' \bmod m.
\end{gather*}
By periodicity of~$S$, for all $v \in [L^*, U-1]$ we have
\begin{equation*}
v \in S \quad\text{if and only if}\quad v + (L'-L^*) \in S
\end{equation*}
and therefore
\begin{equation*}
\card{([L^*, U-1] \cap S)} =
\card{([L', U'-1] \cap S)} =
r_j'.
\end{equation*}
\end{itemize}
Let us sum up the results above.
Due to the semantics of the quantifier $\exists^{\geq \avarter}\avarbis$,
the constraint on the variable~$z$ is equivalent to the following one:
\begin{equation*}
z \le 
p_j \cdot \frac{L^* - L}{m}
+
r_j',
\end{equation*}
which is the same as $m \cdot z \le p_j \cdot (L^* - L) + m \cdot r_j'$.
It remains to return to the original terms, ``undoing'' the variable
assignment~$\aeval$. Observe that
\begin{equation*}
L^* - L = (U - L) - (U - L^*) = (U - L) - (U' - L').
\end{equation*}
In the constraint, instead of $U-L$ we write $t_j' - (t_{j-1}' + 1)$,
and the value of $U' - L'$ can be computed from
$r(t_{j-1}')$ and $r(t_j')$ using simple arithmetic.
Putting everything together, we obtain
\begin{equation*}
m \cdot z \le p_j \cdot 
(t_j' - t_{j-1}' - 1 - (U' - L'))
+ m \cdot r_j'.
\end{equation*}
Since $U'-L' \in [0, m-1]$ and $r_j' \in [0, p_j]$,
the following bounds hold:
\begin{align*}
p_j \cdot ( -1 - (U' - L')) + m \cdot r_j' &\le m \cdot r_j' \le m^2, \\
p_j \cdot ( -1 - (U' - L')) + m \cdot r_j' &\ge - p_j \cdot m \ge - m^2.
\end{align*}
This completes the proof.
\end{claimproof}

\claimPsiFour*
\begin{claimproof}
    Follows directly from~\Cref{claim:psi4-inf-sol},~\Cref{claim:psi4-remove-cq-eq} and~\Cref{claim:psi4-remove-cq-ineq}, 
    together with simple formulae manipulations.
\end{claimproof}

\claimPsiFive*
\begin{claimproof}
    Let~$i \in [1,o]$, $r \colon Z \to [m]$. 
    Establishing~$\Gamma_{i,r} \land \Aformula_4^{i,r}$ $\fequiv$ $\Gamma_{i,r} \land \Aformula_{5}^{i,r}$ suffices. If~$\Aformula_4^{i,r} = \true$ then $\Aformula_5^{i,r}$ is defined as $\true$ and the equivalence holds. 
    Otherwise, we have 
    \begin{align*}
        \Aformula_4^{i,r} &\egdef
        \exists \avar_2\dots\exists\avar_{\ell} \left(
        \avar \leq \sum_{j=2}^{\ell}\avar_j + \sum_{j=1}^{\ell}c_j \land 
        \bigwedge_{j \in [2,\ell]} m\avar_{j} \leq p_j(t_{j}' - t_{j-1}') + r_j
        \right)\\
        \Aformula_5^{i,r} &\egdef
        m\avar \leq \sum_{j=2}^{\ell}( p_j(t_{j}' - t_{j-1}') + r_j ) + m \cdot \sum_{j=1}^{\ell} c_j.
    \end{align*}
    
    \ProofRightarrow
    It is easy to see that~$\Aformula_5^{i,r}$ is obtained from~$\Aformula_4^{i,r}$ by first multiplying both sides of the inequality~$\avar \leq \sum_{j=2}^{\ell}\avar_j + \sum_{j=1}^{\ell}c_j$ by $m$, and then substituting $m\avar_{j}$ with $p_j(t_{j}' - t_{j-1}') + r_j$.

    \ProofLeftarrow
    Let $\aeval$ be an assignment such that $\aeval \models \Gamma_{i,r} \land \Aformula_5^{i,r}$. We show that $\aeval \models \Aformula_4^{i,r}$.
    First of all, we consider $j \in [2,\ell]$, and aim at showing that~$\aeval(p_j(t_{j}' - t_{j-1}') + r_j))$ is a multiple of~$m$. We recall the definition of~$r_j$, $\underline{u}_j$ and $\overline{u}_j$, as introduced in Step IV:
    \begin{itemize}
        \item $r_j = - p_j \cdot (\overline{u}_j - \underline{u}_j) + m \cdot r_j'$,
        \item
        $\underline{u}_j = r(t_{j-1}')$ and $\overline{u}_j$ is the smallest integer congruent to $r(t_j')$ mod~$m$ and~$>\underline{u}_j$.
    \end{itemize}
    By definition of $r_j$, the term $p_j(t_{j}' - t_{j-1}') + r_j$ is equivalent to $p_j(t_{j}' - t_{j-1}' -\overline{u}_j + \underline{u}_j) + m \cdot r_j'$.
    Since $\aeval \models \Gamma_{i,r}$, we have $\underline{u}_j \equiv_m \aeval(t_{j-1}')$ and $\overline{u}_j \equiv_m \aeval(t_{j}')$. 
    From axioms of modular arithmetic, $\aeval(t_{j}') - \aeval(t_{j-1}') -\overline{u}_j + \underline{u}_j \equiv_m 0$, and thus $p_j(\aeval(t_{j}') - \aeval(t_{j-1}') -\overline{u}_j + \underline{u}_j) + m \cdot r_j' \equiv_m 0$, which allows us to conclude that $\aeval(p_j(t_{j}' - t_{j-1}') + r_j))$ is a multiple of $m$.
    Therefore, for every~$j \in [2,\ell]$, there is $v_j \in \Zed$ such that 
    $m \cdot v_j = \aeval(p_j(t_{j}' - t_{j-1}') + r_j))$.
    Let~$\avar_2,\dots,\avar_\ell$ be fresh variables.
    We consider the assignment~$\aeval[v_2/\avar_2,\dots,v_\ell/\avar_\ell]$ that updates~$\aeval$ by assigning~$v_j$ to the variable~$\avar_j$, for every $j \in [2,\ell]$. We have, 
    \begin{center}
    $\aeval[v_2/\avar_2,\dots,v_\ell/\avar_\ell] \models m \avar \leq \sum_{j=2}^{\ell}m\avar_j + m\cdot\sum_{j=1}^{\ell}c_j \land 
    \bigwedge_{j \in [2,\ell]} m\avar_{j} = p_j(t_{j}' - t_{j-1}') + r_j$.
    \end{center}
    Divide both side of the leftmost inequality by $m \in \Nat$, and weaken the equalities of the form $m\avar_{j} = p_j(t_{j}' - t_{j-1}') + r_j$
    to inequalities of the form~$m\avar_{j} \leq p_j(t_{j}' - t_{j-1}') + r_j$.
    We obtain 
    \begin{center}
        $\aeval[v_2/\avar_2,\dots,v_\ell/\avar_\ell] \models \avar \leq \sum_{j=2}^{\ell} \avar_j + \cdot\sum_{j=1}^{\ell}c_j \land 
    \bigwedge_{j \in [2,\ell]} m\avar_{j} \leq p_j(t_{j}' - t_{j-1}') + r_j$.
    \end{center}
    By definition of the existential quantifier, 
    \begin{center}
    $\aeval \models \exists \avar_2 \dots \exists \avar_\ell \left(\avar \leq \sum_{j=2}^{\ell} \avar_j + \cdot\sum_{j=1}^{\ell}c_j \land 
    \bigwedge_{j \in [2,\ell]} m\avar_{j} \leq p_j(t_{j}' - t_{j-1}') + r_j\right)$.
    \end{center}
    That is, $\aeval \models \Aformula_4^{i,r}$.
\end{claimproof}

\section{Missing proofs from~\Cref{s:summary}}
\label{section:complexity-analysis}
In this appendix, we provide the computational analysis on the parameters~$\linterms{.}$, $\homterms{.}$ and~$\fmod{.}$, of the formula obtained form the elimination of the quantifier~$\exists^{\geq \avar}\avarbis$ via the procedure of~\Cref{section:quantifier-elimination}.

Let~$\aformula$ be a quantifier-free formula, and let $d = \card{\freevars{\aformula}}$. Consider the formula~$\Aformula_5$ obtained by performing the quantifier-elimination procedure of~\Cref{section:quantifier-elimination} on the formula~$\exists^{\geq \avar}\avarbis\, \aformula$. 

The following lemma restates~\Cref{lemma:bound-quantifier-elimination}
by expressing the bounds on~$\Aformula_5$ explicitly.
\begin{lemma}
    \label{theorem:bound-quantifier-elimination}
    The following bounds are established for~$\Aformula_5$:
    \begin{itemize}
        \setlength{\itemsep}{4pt}
        \item $\card{\fmod{\Aformula_5}} = \{m\}$ with $m = k \cdot \lcm{(\fmod{\aformula})}$ and $k \leq \BigO{\norminf{\homterms{\aformula}}^{\card{\homterms{\aformula}}}}$.
        \item $\card{\linterms{\Aformula_5}}$ and $\card{\homterms{\Aformula_5}}$ are bounded by $(m \cdot \card{\linterms{\aformula}})^{\BigO{d}},$
        \item $\norminf{\linterms{\Aformula_5}} \leq \BigO{ m^2 \cdot \card{\linterms{\aformula}} \cdot \norminf{\linterms{\aformula}}}$,
        \item $\norminf{\homterms{\Aformula_5}} \leq \BigO{ m^2 \cdot \card{\linterms{\aformula}} \cdot \norminf{\homterms{\aformula}}}$.
    \end{itemize}
\end{lemma}

\begin{proof}
    First of all, from~\Cref{lemma:make-modulo-simple}, we notice that translating every modulo constraint appearing in~$\aformula$ into simple modulo constraints does not change the sets~$\linterms{\aformula}$, $\homterms{\aformula}$ and~$\fmod{\aformula}$. Therefore, assume~$\aformula$ to be a Boolean combination of linear inequalities and simple modulo constraints.
    Let~$k$ be the \emph{lcm} of the absolute values of all coefficients in~$y$ appearing in $\homterms{\aformula}$.
    We have $k \leq \norminf{\homterms{\aformula}}^{\card{\homterms{\aformula}}}$.
    The first step essentially multiplies every term in~$\aformula$ by $k$, producing the formula~$\Aformula_1$ with bounds 
    \begin{itemize}
        \setlength{\itemsep}{3pt}
        \item $\card{\linterms{\Aformula_1}} = \card{\linterms{\aformula}}$ and $\norminf{\linterms{\Aformula_1}} \leq k\norminf{\linterms{\aformula}}$,
        \item $\card{\homterms{\Aformula_1}} = \card{\homterms{\aformula}}$ and $\norminf{\homterms{\Aformula_1}} \leq k\norminf{\homterms{\aformula}}$,
        \item $\fmod{\Aformula_1} = \{kq \mid q \in \fmod{\aformula}\}$.
    \end{itemize}
    Let~$T$ be the set of all $y$-free terms~$t$ such that $t$, $y - t$ or $-y + t$ belong to~$\linterms{\Aformula_1}$.
    So, $\card{T} \leq \card{\linterms{\Aformula_1}}$, 
    $\norminf{T} \leq \norminf{\linterms{\Aformula_1}}$
    and all coefficients of variables in terms of~$T$ are bounded by~$\norminf{\homterms{\Aformula_1}}$.
    Let~$m = \lcm(\fmod{\Aformula_1}) = k \cdot \lcm{(\fmod{\aformula})}$.
    In the second step of the procedure, the orderings introduce terms $\aterm \lhd \aterm'$, where $\lhd \in \{<,=\}$ and $\aterm,\aterm' \in T \cup \{0\}$. So, at most $(\card{T} \cup \{0\})^2$ new terms are introduces, increasing $\linterms{.}$ and $\homterms{.}$ quadratically in cardinality. The magnitude of coefficients and constants doubles. 
    Simple modulo constraints of the form $x \equiv_m r$ are also introduced. 
    Because of this, the formula~$\Aformula_2$ produced in the second step of the procedure has the following bounds:
    \begin{itemize}
        \setlength{\itemsep}{3pt}
        \item $\card{\linterms{\Aformula_2}} \leq (\card{\linterms{\Aformula_1}} + 1)^2 + \card{\linterms{\Aformula_1}}$ and $\norminf{\linterms{\Aformula_2}} \leq 2\norminf{\linterms{\Aformula_1}}$,
        \item $\card{\homterms{\Aformula_2}} = (\card{\homterms{\Aformula_1}}+1)^2 + \card{\homterms{\Aformula_1}}$ and $\norminf{\homterms{\Aformula_2}} \leq 2\norminf{\homterms{\Aformula_1}}$,
        \item $\fmod{\Aformula_2} = \{m\} \cup \fmod{\Aformula_1}$.
    \end{itemize}
    To study the bounds on the formula~$\Aformula_5$, analysing the bounds obtained from the third and fourth steps of the procedure is unnecessary. Indeed, we recall that~$\Aformula_5$ is defined as 
    \begin{center}
        $\Aformula_5 = \bigvee_{i \in [1,o]}\bigvee_{r\colon Z\to [m]} (\Gamma_{i,r} \land \Aformula_5^{i,r})$,
    \end{center}
    where every $\Gamma_{i,r}$ is a conjunction of simple modulo constraints of the form $z \equiv_m r$ and linear inequalities from~$\Aformula_2$, and every $\Aformula_5^{i,r}$ is either $\true$ or a formula of the form 
    \begin{center}
        $m\avar \leq \sum_{j=2}^{\ell}( p_j(t_{j}' - t_{j-1}') + r_j ) + m \cdot \sum_{j=1}^{\ell} c_j$
    \end{center}
    where
    $\ell \leq \card{T}+1 \leq \card{\linterms{\Aformula_1}}+1$, for every~$j \in [1,\ell$] $c_j \in \{0,1\}$, and for every ${j \in [2,\ell]}$, ${p_j \in [0,m]}$ and $\abs{r_j} \in[-m^2,m^2]$ (see~\Cref{claim:psi4-remove-cq-ineq}) and the terms~$\aterm_j'$ 
    and~$\aterm_{j-1}'$ belongs to $T$.
    This implies that variable coefficients in~$\Aformula_5^{i,r}$ are bounded (in absolute value) by ${m \cdot (2 \cdot \ell \cdot \norminf{\homterms{\Aformula_1}}+1)}$, whereas the constant term is bounded by 
    $2 \cdot \ell \cdot m \cdot \norminf{\linterms{\Aformula_1}} + \ell \cdot m^2 + m$, again in absolute values.
    By recalling that the number of disjunctions of~$\Aformula_5$ is $m^do \leq m^d(2(\card{(T \cup \{0\})}^2 + 1)^d$
    (see \Cref{lemma:bound-number-of-orderings} for the bound on~$o$), we derive 
    \begin{itemize}
        \setlength{\itemsep}{3pt}
        \item $\fmod{\Aformula_5} = \{m\}$, and so $\norminf{\fmod{\Aformula_5}} \leq \lcm(\fmod{\aformula}) \cdot \norminf{\homterms{\aformula}}^{\card{\homterms{\aformula}}}$,
        \item $
            \begin{aligned}[t]
            \card{\linterms{\Aformula_5}} &\leq (\card{\linterms{\Aformula_1}} + 1)^2 + (2(\card{\linterms{\Aformula_1}+1)}^2 + 1)^d m^d\\
            &\leq \BigO{m^d \cdot \card{\linterms{\aformula}}^{2d}}
            \end{aligned}
            $ 
        \item $
        \begin{aligned}[t]
        \norminf{\linterms{\Aformula_5}} &\leq 2 \cdot \ell \cdot m \cdot \norminf{\linterms{\Aformula_1}} + \ell \cdot m^2 + m\\  
        & \leq \BigO{m^2 \cdot \card{\linterms{\aformula}} \cdot \norminf{\linterms{\aformula}}}
        \end{aligned}$
        \item $
        \begin{aligned}[t]
        \card{\homterms{\Aformula_5}} &\leq (\card{\homterms{\Aformula_1}} + 1)^2 + (2(\card{\linterms{\Aformula_1}+1)}^2 + 1)^d m^d\\
        &\leq \BigO{m^d \cdot \card{\linterms{\aformula}}^{2d}}
        \end{aligned}
        $ 
        \item $
        \begin{aligned}[t]
        \norminf{\homterms{\Aformula_5}} &\leq m \cdot (2\cdot\ell\cdot\norminf{\homterms{\Aformula_1}}+1)\\
        &\leq \BigO{ m^2 \cdot \card{\linterms{\aformula}} \cdot \norminf{\homterms{\aformula}}}.
        \end{aligned}$
        \vspace{-0.6cm}
    \end{itemize}

\end{proof}

\section{Missing proofs from \Cref{section:elimination-treshold-quantifiers}}

In the lemma below, we recall that $d = \vars{\Aformula_5^i,r}$
and that $\Aformula_5^{i,r}$ has the following form (see~\Cref{equation:5ir-form} in the body of the paper):
\[
    \textstyle
    m\avar \leq \sum_{j=2}^{\ell}( p_j(t_{j}' - t_{j-1}') + r_j ) + m \cdot \sum_{j=1}^{\ell} c_j
\]

\LemmaSimplifyingThreshold*

\begin{proof}
    Consider the set of $e+\ell$ terms $T' = \{\aterm_j' - \aterm_{j-1}' \mid j \in [2,\ell]\} \cup [0,e]$.
    Note that $e \leq 2 \cdot c \cdot \norminf{\Aformula_5^{i,r}}$, and so $\norminf{T'} \leq \BigO{c \cdot \norminf{\Aformula_5^{i,r}}}$. 
    Applying~\Cref{lemma:bound-number-of-orderings}, 
    we compute a set 
    $\{\ordering_1',\dots,\ordering_o'\}$ of orderings for~$T'$ such that $\psi \egdef \bigvee_{k \in [1,o]} \ordering_k'$ is a tautology and $o' = \mathcal O((e+\ell)^{2 d})$.
    Since $\ordering_k'$ is an ordering for $T'$, for all $k \in [1,o']$ and $t',t'' \in T'$, 
    exactly one of the entailments 
    $\ordering_k' \models t' < t''$,
    $\ordering_k' \models t' = t''$ or 
    $\ordering_k' \models t' > t''$ holds.
    
    We iterate over all $k \in [1,o']$, 
    at each step generating a formula $\psi_k$ that satisfies \begin{equation}
        \label{equation:equiv5ir-psi}
        \Gamma_{i,r} \land \Aformula_5^{i,r}\substitute{\avarter}{c} \fequiv (\Gamma_{i,r} \land \psi_k).
    \end{equation}
    At the end of the process, the formula $\psi_{o'}$ is the formula $\aformulater_{i,r}$ required by the lemma.
    Let $\psi_0 = \bigvee_{k \in [1,o]} (\Aformula_5^{i,r}\substitute{\avarter}{c} \land \ordering_k')$.
    Since $\psi$ is a tautology and $\psi_0 \fequiv \Aformula_5^{i,r}\substitute{\avarter}{c} \land \psi$, the formula $\psi_0$ satisfies the equivalence in~\eqref{equation:equiv5ir-psi}. 
    Let $n = \card{T'} = \ell + e$. For all $k \in [1,o']$, suppose 
    \[
        \ordering_k' = b_1 \lhd_1 b_2 \land \dots \land b_{n-1} \lhd_{n-1} b_n    
    \]
    where $\{b_1,\dots,b_n\} = T'$ and $\{\lhd_1,\dots,\lhd_{n-1}\} \subseteq \{<,=\}$. 
    We inductively assume that $\aformulabis_{k-1}$ satisfies the equivalence~\eqref{equation:equiv5ir-psi},
    and we compute $\aformulabis_k$ following the cases below. Notice that checking which of the cases is satisfied by $\ordering_k'$ can be done in linear time with respect to $\abs{\ordering_k'}$, by simply scanning the ordering.
    \begin{description}
        \item[case: $\ordering_k'$ does not respect the order $0 < 1 < \dots < e$.] 
        Then, $\ordering_k'$ is unsatisfiable and $\psi_k$ is obtained from $\psi_{k-1}$ by removing the disjunct $\Aformula_5^{i,r}\substitute{\avarter}{c} \land \ordering_k'$. 
        Since $\psi_{k-1}$ satisfies the equivalence~\eqref{equation:equiv5ir-psi}, so does~$\psi_k$.
        \item[case: $\ordering_k' \models \aterm_j' - \aterm_{j-1}' < 0$, for some $j \in {[2,\ell]}$].
        Since $\Aformula_5^{i,r} \models \aterm_{j-1}' < \aterm_j'$, the formula $\Aformula_5^{i,r}\substitute{\avarter}{c} \land \ordering_k'$ is unsatisfiable.
        Again, $\psi_k$ is obtained from $\psi_{k-1}$ by removing the disjunct $\Aformula_5^{i,r}\substitute{\avarter}{c} \land \ordering_k'$, and~$\psi_k$ satisfies the equivalence~\eqref{equation:equiv5ir-psi}.
        \item[otherwise,] for every $j \in [2,\ell]$ there is $i_j \in [0,e]$ such that either $\ordering_k' \models i_j = \aterm_j' - \aterm_{j-1}'$ or $i_j = e$ and $\ordering_k \models i_j < \aterm_j' - \aterm_{j-1}'$.
        By simply parsing of the ordering, can find all 
        the~$i_j$ in time~$\BigO{\abs{\ordering_k'}}$.
        Now, if $e \leq \sum_{j = 2}^\ell p_j \cdot i_j$ does not hold, 
        then the formula $\Aformula_5^{i,r}\substitute{\avarter}{c} \land \ordering_k'$ is unsatisfiable and, 
        as in the previous cases, we define $\psi_k$ from $\psi_{k-1}$ by removing this disjunct. 
        We obtain a formula that satisfies the equivalence~\eqref{equation:equiv5ir-psi}.
        Otherwise, let $\aformulater = \bigwedge_{j \in [2,\ell]} \aterm_j' - \aterm_{j-1} \geq i_j$.
        By definition, $\ordering_k' \models \aformulater$ 
        and $\aformulater \models \Aformula_5^{i,r}\substitute{\avarter}{c}$.
        Let $\aformulabis_k$ be the formula obtained from $\aformulabis_{k-1}$ by replacing the disjunct $\Aformula_5^{i,r}\substitute{\avarter}{c} \land \ordering_k'$ by the formula~$\aformulater$. 
        Notice that $\aformulabis_{k-1} \models \aformulabis_k$,
        directly from $\ordering_k' \models \aformulater$.
        We show that $\aformulabis_k$ satisfies the equivalence~\eqref{equation:equiv5ir-psi}. 
        
        \ProofRightarrow 
        Let $\aeval$ be an assignment such that 
        $\aeval \models \Gamma_{i,r} \land \Aformula_5^{i,r}$.
        Since $\aformulabis_{k-1}$ satisfies the equivalence~\eqref{equation:equiv5ir-psi}, 
        we have $\aeval \models \Gamma_{i,r} \land \aformulabis_{k-1}$. By $\aformulabis_{k-1} \models \aformulabis_k$, we derive 
        $\aeval \models \Gamma_{i,r} \land \aformula_{k}$.

        \ProofLeftarrow Let $\aeval$ be an assignment such that $\aeval \models \Gamma_{i,r} \land \aformulabis_k$. 
        If $\aeval$ satisfies a disjunct of $\aformulabis_k$ that is different from $\aformulater$, then $\aeval \models \aformulabis_{k-1}$ and, since $\aformulabis_{k-1}$ satisfies the equivalence~\eqref{equation:equiv5ir-psi}, $\aeval \models \Gamma_{i,r} \land \Aformula_5^{i,r}$.
        Otherwise, $\aeval \models \aformulater$ and, by $\aformulater \models \Aformula_5^{i,r}\substitute{\avarter}{c}$ we deduce that $\aeval \models \Gamma_{i,r} \land \Aformula_5^{i,r}\substitute{\avarter}{c}$.
    \end{description}
    As already said, the formula $\aformulater_{i,r} \egdef \aformulabis_{o'}$. The formula $\aformulater_{i,r}$ satisfies all the expected properties.
    In particular, $I \subseteq [0,e]^\ell$ holds by definition of the various $i_j$ in the third case of the procedure, and $\card{I} \leq \BigO{(e+\ell)^{2d}}$ holds from the bound on the number $o'$ of disjuncts of $\aformulabis$.
    By~\Cref{lemma:bound-number-of-orderings}, 
    computing the initial formula $\aformulabis$ can be done in time~$(e+\ell)^{\mathcal O (d)} \log \norminf{T'}^{\mathcal O (1)}$.
    Similarly, the case analysis on the disjuncts of $\aformulabis$ has a overall running time that is linear in $\abs{\aformulabis} \leq (e+\ell)^{\mathcal O (d)} \log \norminf{T'}^{\mathcal O (1)}$.
\end{proof}

\ClaimThresholdQEProcedure*
\begin{claimproof}
  By definition of $\Aformula_6^c$ 
  and~\Cref{lemma:simplifying-threshold}, 
  $\Aformula_5 \fequiv \Aformula_6^c$. 
  Then, the claim follows from the chain of claims ``$\Aformula_{i} \fequiv \Aformula_{i+1}$'' starting from~\Cref{claim:psi1} and ending with~\Cref{claim:psi5}.
\end{claimproof}

\begin{restatable}{lemma}{LemmaBoundQuantifierThreshold}
    \label{lemma:bound-quantifier-elimination-threshold}
    The following bounds are established for~$\Aformula_6^c$:
    \begin{itemize}
        \setlength{\itemsep}{3pt}
        \item $\fmod{\Aformula_6^c} = \{m\}$ with $m = k \cdot \lcm{(\fmod{\aformula})}$ and $k \leq \norminf{\homterms{\aformula}}^{\card{\homterms{\aformula}}}$,
        \item $\card{\homterms{\Aformula_6^c}} \leq \BigO{\card{\homterms{\aformula}}^2}$ 
        and $\norminf{\homterms{\Aformula_6^c}} \leq \BigO{k \cdot \norminf{\homterms{\aformula}}}$,
        \item $\card{\linterms{\Aformula_6^c}} \leq \BigO{c \cdot m^2 \cdot \card{\linterms{\aformula}}^2}$ and $\norminf{\linterms{\Aformula_6^c}} \leq \BigO{k \cdot c \cdot m^2 \cdot \card{\linterms{\aformula}} \cdot \norminf{\linterms{\aformula}}}$.
    \end{itemize}
\end{restatable}

\begin{proof}
    Without loss of generality, we assume $\card{\homterms{\aformula}}$, $\card{\linterms{\aformula}}$, $\norminf{\homterms{\aformula}}$ and $\norminf{\linterms{\aformula}}$ to be at least $1$.
    We also assume $c$ to be at least $1$, as otherwise the formula $\exists^{\geq c} \avarbis\, \aformula$ is trivially true. 
    These assumptions hide constant factors in the exponent. 
    We recall the bounds (as in~\Cref{theorem:bound-quantifier-elimination}) on the formula~$\Aformula_1$ obtained after performing the normalisation of the coefficients of~$\avarbis$, 
    as described in Step I of~\Cref{section:quantifier-elimination}. We have
    \begin{itemize}
        \setlength{\itemsep}{3pt}
        \item $\card{\linterms{\Aformula_1}} = \card{\linterms{\aformula}}$ and $\norminf{\linterms{\Aformula_1}} \leq k\norminf{\linterms{\aformula}}$,
        \item $\card{\homterms{\Aformula_1}} = \card{\homterms{\aformula}}$ and $\norminf{\homterms{\Aformula_1}} \leq k\norminf{\homterms{\aformula}}$.
    \end{itemize}
    where $k \leq \norminf{\homterms{\aformula}}^{\card{\homterms{\aformula}}}$
    is the \emph{lcm} of all coefficients of~$\avarbis$ appearing in linear inequalities.

    We recall that the formula~$\Aformula_6^c$ is defined as~$\bigvee_{i \in [1,o]}\bigvee_{r\colon Z\to [m]} (\Gamma_{i,r} \land \Aformula_6^{i,r})$, where ${Z = \freevars{\aformula}}$, $m = k \cdot \lcm(\fmod{\aformula})$,
    $\Gamma_{i,r} = \ordering_i \land (\bigwedge_{\avarfour \in Z} \avarfour \equiv_{m} r(\avarfour))$, 
    where~$\ordering_i$ is an ordering on the set of terms~$T \cup \{0\}$, 
    and~$\Aformula_6^{i,r}$ is either~$\true$ or 
    of the form (see~\Cref{lemma:simplifying-threshold})
    \begin{center}
    $\bigvee_{(i_2,\dots,i_\ell) \in I} \bigwedge_{j \in [2,\ell]} \aterm_j' - \aterm_{j-1}' \geq i_j$. 
    \end{center}
    Here, $\ell \leq \card{T} + 1$ and, for all $j \in [2,\ell]$, $\aterm_j',\aterm_{j-1}' \in T$ 
    and $i_j \in [0,e]$ where $e \leq m \cdot (c + (m+1) \cdot \ell)$.
    Moreover, $\card{T} \leq \card{\linterms{\Aformula_1}}$, $\norminf{T} \leq \norminf{\linterms{\Aformula_1}}$
    and all coefficients of variables in terms of~$T$ are bounded by~$\norminf{\homterms{\Aformula_1}}$.
    So, when accounting for all orderings $(\ordering_i)_{i \in [1,o]}$ and all inequalities of the form $\aterm_j' - \aterm_{j-1}' \geq i_j$, the formula~$\Aformula_6^c$ contains $(e+1)\cdot (\card{T}+1)^2$ inequalities. However, the set~$\homterms{\Aformula_6^c}$ is only quadratic on the size of the set of homogeneous terms built from pairs of terms in~$T \cup \{0\}$, as we do not account for the natural numbers~$i_j$. Since the terms in $T$ are constructed by removing $\avarbis$ from terms in $\linterms{\Aformula_1}$, $\card{\homterms{\Aformula_6^c}}$ is quadratic on $\card{\homterms{\Aformula_1}}$. 
    The magnitude of the coefficients of the variables in linear inequalities of~$\Aformula_6^c$ doubles with respect to~$\norminf{T}$, whereas the magnitude of the constants is bounded by~$2\norminf{T} + e$.
    Lastly, every modulo constraint in~$\Aformula_6^{c}$ is of the form~$\avarfour \equiv_{m} r(\avarfour)$, and
    thus~$\fmod{\Aformula_6^c} = \{m\}$. Overall, the following bounds are derived:
    \begin{itemize}
        \setlength{\itemsep}{3pt}
        \item $\card{\fmod{\Aformula_6^c}} = \{m\}$ with $m = k \cdot \lcm{(\fmod{\aformula})}$ and $k \leq \norminf{\homterms{\aformula}}^{\card{\homterms{\aformula}}}$,
        \item $\card{\homterms{\Aformula_6^{c}}} \leq (\card{\homterms{\aformula}}+1)^2 \leq 4 \cdot \card{\homterms{\aformula}}^2$,
        \item $\norminf{\homterms{\Aformula_6^{c}}} \leq 2\cdot k \cdot \norminf{\homterms{\aformula}}$,
        \item $\card{\linterms{\Aformula_6^{c}}} \leq (e+1) \cdot (\card{\linterms{\aformula}}+1)^2 \leq 20 \cdot c \cdot m^2 \cdot \card{\linterms{\aformula}}^3$,
        \item $\norminf{\linterms{\Aformula_6^{c}}} \leq 2 \cdot k \cdot \norminf{\linterms{\aformula}} + e \leq  
        6 \cdot k \cdot c \cdot m^2 \cdot \card{\linterms{\aformula}} \cdot \norminf{\linterms{\aformula}}$,
    \end{itemize}
    where we recall that we are assuming $\card{\homterms{\aformula}}$, $\card{\linterms{\aformula}}$, $\norminf{\homterms{\aformula}}$, $\norminf{\linterms{\aformula}}$,
    $c$ $\geq$ $1$.
\end{proof}

\LemmaBoundQuantifierThresholdSimplified*

\begin{proof}
This is a simple consequence of~\Cref{lemma:bound-quantifier-elimination-threshold}.
\end{proof}

\LemmaBoundQuantTreshMultiple*

\begin{proof} 
    Recall that the standard first-order quantifier~$\exists \avarbis$ is equivalent to~$\exists^{\geq 1} \avarbis$. Therefore, without loss of generality, we can assume~$\aformula$ to only contain threshold counting quantifiers. For simplicity, we also assume $\card{\homterms{\aformula}}$, $\card{\linterms{\aformula}}$, $\norminf{\homterms{\aformula}}$ and $\norminf{\linterms{\aformula}}$ to be at least $1$. This hides constant factors in the exponent of the bounds that we derive.
    Le us introduce some shortcuts.
    \begin{itemize}
        \item Let $d$ be the the quantifier-depth of~$\aformula$, 
        \item let $B$ be $2$ plus~$\card{\fmod{\aformula}}$, plus the number of Boolean connectives in~$\aformula$, and 
        \item let $\bar{c}$ be the maximal integer such that~$\exists^{\geq \bar{c}}\avarbis$ occurs in~$\aformula$. 
    \end{itemize}
    We show the following bounds for~$\Aformula$, sharpening the ones in the statement of the lemma.
    \begin{itemize}
        \setlength{\itemsep}{3pt}
        \item $\card{\homterms{\Aformula}} \leq A_d \egdef
            (4 \cdot B)^{2^d-1} \cdot \card{\homterms{\aformula}}^{2^{d}},$
        \item $\card{\fmod{\Aformula}} \leq B$,
        \item $\norminf{\homterms{\Aformula}} \leq C_d \egdef 2^{(2A_d)^{d}-1}\norminf{\homterms{\aformula}}^{(2A_d)^d}$,
        \item $\norminf{\fmod{\Aformula}} \leq D_d \egdef (C_d)^{(B^d-1)} \cdot \lcm(\fmod{\aformula})^{B^d}$,
        \item $\card{\linterms{\Aformula}} \leq 
            E_d \egdef (20 \cdot \bar{c} \cdot B \cdot {D_d}^2)^{3^d-1} \cdot \card{\linterms{\aformula}}^{3^d}$,
        \item $\norminf{\linterms{\Aformula}} \leq F_d \egdef (6 \cdot \bar{c} \cdot {D_d}^3 \cdot E_d)^{d} \norminf{\linterms{\aformula}}$.
    \end{itemize}
    Notice that~$A_d$, $C_d$, $D_d$, $E_d$ and $F_d$ are monotonous in $d$.
    Moreover, notice that $B \leq \BigO{\abs{\aformula}}$.
    The proof is by induction on the quantifier-depth of~$\aformula$. 
    The base case fore $d = 0$, i.e.~$\aformula$ quantifier-free, is trivial.
    For the induction step, let $S = \{\exists^{\geq c_1}\avarbis_1\,\aformulabis_1,\dots,\exists^{\geq_k}\avarbis_n\,\aformulabis_n\}$ be a minimal family of formulae such that~$\aformula$ is a Boolean combination of formulae from $S$. Notice that $n \leq B$.
    Let $j \in [1,n]$.
    The quantifier-depth of $\aformulabis_j$ is at most $d-1$.
    We apply the quantifier elimination procedure on~$\aformulabis_j$, obtaining the formula~$\Aformula_j$. By induction hypothesis,
    \begin{itemize}
        \setlength{\itemsep}{3pt}
        \item $\card{\homterms{\Aformula_j}} \leq 
        A_{d-1} =(4 \cdot B)^{2^{d-1}-1} \cdot \card{\homterms{\aformula}}^{2^{d-1}},$
        \item $\card{\fmod{\Aformula_j}} \leq B$,
        \item $\norminf{\homterms{\Aformula_j}} \leq C_{d-1} = 2^{(2A_{d-1})^{d-1}-1}\norminf{\homterms{\aformula}}^{(2A_{d-1})^{d-1}}$,
        \item $\norminf{\fmod{\Aformula_j}} \leq D_{d-1} = (C_{d-1})^{(B^{d-1}-1)} \cdot \lcm(\fmod{\aformula})^{B^{d-1}}$,
        \item $\card{\linterms{\Aformula_j}} \leq E_{d-1} =
        (20 \cdot \bar{c} \cdot B \cdot {D_{d-1}}^2)^{(3^{d-1}-1)} \cdot \card{\linterms{\aformula}}^{3^{d-1}}$,
        \item $\norminf{\linterms{\Aformula_j}} \leq F_{d-1} = (6 \cdot \bar{c} \cdot {D_{d-1}}^3 \cdot E_{d-1})^{d-1} \norminf{\linterms{\aformula}}$.
    \end{itemize}
    For~$j \in [1,k]$, we consider every formula $\exists^{\geq c}\avarbis_j \,\Aformula_j$ and perform 
    the quantifier elimination procedure for threshold counting quantifiers, obtaining a formula $\widetilde{\Aformula}_j$. 
    From~\Cref{lemma:bound-quantifier-elimination-threshold} (see the proof of this lemma for the exact bounds) 
    we have 
    \begin{itemize}
        \setlength{\itemsep}{3pt}
        \item $\card{\fmod{\widetilde{\Aformula}_j}} = \{m\}$ with $m = k \cdot \lcm{(\fmod{\Aformula_j})}$ and $k \leq \norminf{\homterms{\Aformula_j}}^{\card{\homterms{\Aformula_j}}}$,
        \item $\card{\homterms{\widetilde{\Aformula}_j}} \leq 4 \cdot \card{\homterms{\Aformula_j}}^2$,
        \item $\norminf{\homterms{\widetilde{\Aformula}_j}} \leq  2 \cdot k \cdot \norminf{\homterms{\Aformula_j}} \leq 2 \cdot \norminf{\homterms{\Aformula_j}}^{2\card{\homterms{\Aformula_j}}}$,
        \item $\card{\linterms{\widetilde{\Aformula}_j}} \leq 20 \cdot \bar{c} \cdot m^2 \cdot \card{\linterms{\Aformula_j}}^3$,
        \item $\norminf{\linterms{\widetilde{\Aformula}_j}} \leq  6 \cdot k \cdot \bar{c} \cdot m^2 \cdot \card{\linterms{\Aformula_j}} \cdot \norminf{\linterms{\Aformula_j}} \leq 6 \cdot \bar{c} \cdot m^3 \cdot \card{\linterms{\Aformula_j}} \cdot \norminf{\linterms{\Aformula_j}}$,
    \end{itemize}
    We derive:
    \begin{itemize}
        \item $\card{\homterms{\widetilde{\Aformula}_j}} \leq 4 \cdot ((4 \cdot B)^{2^{d-1}-1} \cdot \card{\homterms{\aformula}})^{2^{d-1}})^2 \leq 
        B^{2^d-2}4^{2^d-1}\card{\homterms{\aformula}}^{2^d} = \frac{A_d}{B}$.
        \item $
        \begin{aligned}[t]
        \norminf{\homterms{\widetilde{\Aformula}_j}} &\leq  2 (2^{(2A_{d-1})^{d-1}-1}\norminf{\homterms{\aformula}}^{(2A_{d-1})^{d-1}})^{2A_{d-1}}\\
        &\leq 2^{(2A_{d-1})^{d}-1}\norminf{\homterms{\aformula}}^{(2A_{d-1})^d} \leq 2^{(2A_{d})^{d}-1}\norminf{\homterms{\aformula}}^{(2A_{d})^d} = C_d.
        \end{aligned}$
    \end{itemize}
    Notice that $k \leq \norminf{\homterms{\Psi_j}}^{\card{\homterms{\Psi_j}}} \leq C_d$ and that $\lcm(\fmod{\Aformula_j}) \leq \norminf{\fmod{\Aformula_j}}^B$.
    \begin{itemize}
        \item $
        \begin{aligned}[t]
        m &\leq  \lcm{(\fmod{\Aformula_j})} \cdot C_d \leq C_d \cdot \norminf{\fmod{\Aformula_j}}^B\\
        &\leq C_d \cdot ((C_{d-1})^{(B^{d-1}-1)} \cdot \lcm(\fmod{\aformula})^{B^{d-1}})^B\\
        &\leq (C_d)^{(B^d-B+1)} \cdot \lcm(\fmod{\aformula})^{B^d}
        \leq (C_d)^{(B^d-1)} \cdot \lcm(\fmod{\aformula})^{B^d} = D_d,
        \end{aligned}$
        \item[] 
        where we recall that we assume $B \geq 2$.
        Hence, $\card{\fmod{\widetilde{\Aformula}_j}} = 1$ and $\norminf{\fmod{\widetilde{\Aformula}}_j} \leq D_d$.
        \item $
        \begin{aligned}[t]
            \card{\linterms{\widetilde{\Aformula}_j}} &\leq 20 \cdot \bar{c} \cdot m^2 \cdot \card{\linterms{\Aformula_j}}^3 \leq (20 \cdot \bar{c} \cdot m) \cdot ((20 \cdot \bar{c} \cdot B \cdot {D_d}^2)^{(3^{d-1}-1)} \cdot \card{\linterms{\aformula}}^{3^{d-1}})^3\\
            &\leq B^{3^{d}-3}(20 \cdot \bar{c} \cdot {D_d}^2)^{3^d-2}\cdot \card{\linterms{\aformula}}^{3^{d}} \leq \textstyle\frac{E_d}{B}.
        \end{aligned}$
        \item $\norminf{\linterms{\widetilde{\Aformula}_j}} \leq  6 \cdot \bar{c} \cdot m^3 \cdot \card{\linterms{\Aformula_j}} \cdot \norminf{\linterms{\Aformula_j}} \leq (6 \cdot \bar{c} \cdot {D_d}^3 \cdot E_d) \cdot (6 \cdot \bar{c} \cdot {D_d}^3 \cdot E_d)^{d-1} \norminf{\linterms{\aformula}} = F_d$.
    \end{itemize}
    For every $j \in [1,n]$, we replace in $\aformula$ 
    each occurrence of the formula~$\exists^{\geq c_j}\avar_j\,\aformulabis_j$ 
    not in the scope of a quantification with the formula~$\widetilde{\Aformula}_j$. We obtain the formula~$\Aformula$ that is a Boolean combination of 
    at most $B$ formulae from $\{\widetilde{\Aformula}_1,\dots,\widetilde{\Aformula}_n\}$. So, $\Aformula$ is quantifier-free.
    We have:
    \begin{itemize}
        \item $\card{\homterms{\Aformula}} = \sum_{j = 1}^n \card{\homterms{\widetilde{\Aformula}_j}} \leq A_d$,
        \item $\card{\fmod{\Aformula}} = \sum_{j = 1}^n \card{\fmod{\widetilde{\Aformula}_j}} \leq B$, 
        \item 
        $\norminf{\homterms{\Aformula}} \leq \max\{\norminf{\homterms{\widetilde{\Aformula}_j}}\mid j \in [1,n]\} \leq C_d$,
        \item 
        ${\norminf{\fmod{\Aformula}} \leq \max\{\norminf{\fmod{\widetilde{\Aformula}_j}}\mid j \in [1,n]\} \leq D_d}$,
        \item
        $\card{\linterms{\Aformula}} = \sum_{j = 1}^n \card{\linterms{\widetilde{\Aformula}_j}} \leq E_d$,
        \item 
        $\norminf{\linterms{\Aformula}} \leq \max\{\norminf{\linterms{\widetilde{\Aformula}_j}}\mid j \in [1,n]\} \leq F_d$.
        \qedhere
    \end{itemize}
\end{proof}

\section{Eliminating the counting quantifier~$\exists^{=\avar}\avarbis$}
\label{subsection:exists-equal-elim}

We adapt the quantifier-elimination procedure of~\Cref{section:quantifier-elimination} in order to directly deal with the counting quantifier~$\exists^{=\avar}\avarbis$. To shorten the presentation, we only provide the (minor) changes to the procedure of~\Cref{section:quantifier-elimination}. Consider~$\Aformula_0 = \exists^{=\avar}\avarbis\,\aformula$, where 
$\aformula$ is quantifier-free.

\subparagraph*{Steps I--III.}
The first two steps of the procedure follow in the same way as described in~\Cref{section:quantifier-elimination}, the only difference being that the counting quantifier~$\exists^{\geq \avar}\avarbis$ is substituted with~$\exists^{=\avar}\avarbis$. The third step is also analogous, but
instead of defining $\Aformula_3^{i,r}$ as 
\begin{center}
$\exists \avar_0\dots\exists\avar_{2\ell} \left(
            \avar \leq \avar_0 + \dots + \avar_{2\ell} \land \bigwedge_{j \in [0,2\ell]} \exists^{\geq \avar_j} \avarbis (\aformulafour_j \land \aformulabis^{i,r}_{\aformulafour_j})\right)$,
\end{center}
we define it as 
\begin{center}
    $\exists \avar_0\dots\exists\avar_{2\ell} \left(
                \avar = \avar_0 + \dots + \avar_{2\ell} \land \bigwedge_{j \in [0,2\ell]} \exists^{= \avar_j} \avarbis (\aformulafour_j \land \aformulabis^{i,r}_{\aformulafour_j})\right)$.
\end{center}
This change of the procedure is to be expected, in view of the difference between the two forms of quantification.

\subparagraph*{Step IV.}
The fourth step of the procedure is updated to deal with the different semantics that~$\exists^{\geq \avar}\avarbis\,\aformulabis$ and $\exists^{= \avar}\avarbis\,\aformulabis$ assume when infinitely many values of $\avarbis$ satisfy $\aformulabis$. 
In the first case, the formula~$\exists^{\geq \avar}\avarbis\,\aformulabis$
is equivalent to~$\true$, this follows from~\Cref{claim:psi4-inf-sol}. As already stated, the formula~$\exists^{= \avar}\avarbis\,\aformulabis$
instead evaluates to~$\false$, which lead us to update~\Cref{claim:psi4-inf-sol} as follows.%
\begin{claim}
    \label{claim:psi4-inf-sol-neg}
    Let $\aformulafour \in \{\avarbis < t_1',\ t_\ell' < \avarbis\}$.
    If $\exists \avarbis\, (\aformulafour \land \aformulabis_{\aformulafour}^{i,r})$ is satisfiable, then $\Gamma_{i,r} \land \Aformula^{i,r}_3$ $\fequiv$ $\false$.
\end{claim}
The procedure then follows as described in~\Cref{section:quantifier-elimination}, using~\Cref{claim:psi4-inf-sol-neg} instead of~\Cref{claim:psi4-inf-sol} to set $\Aformula_4^{i,r}$ to $\false$ when necessary. 
Whenever $\Aformula_4^{i,r}$ is different form $\false$, instead of having the form  
\begin{center}
    $\exists \avar_2\dots\exists\avar_{\ell} \left(
        \avar \leq \avar_2 + \dots + \avar_{\ell} + c_1 + \dots + c_\ell \land 
        \bigwedge_{j \in [2,\ell]} m\avar_{j} \leq p_j(t_{j}' - t_{j-1}') + r_j
        \right)$
\end{center}
(as defined in~\Cref{section:quantifier-elimination}), it is of the form 
\begin{center}
    $\exists \avar_2\dots\exists\avar_{\ell} \left(
        \avar = \avar_2 + \dots + \avar_{\ell} + c_1 + \dots + c_\ell \land 
        \bigwedge_{j \in [2,\ell]} m\avar_{j} = p_j(t_{j}' - t_{j-1}') + r_j
        \right)$.
\end{center}
Again, this update to the procedure only reflects the differences between~$\exists^{\geq \avar}\avarbis\,\aformulabis$ and $\exists^{= \avar}\avarbis\,\aformulabis$.

\subparagraph*{Step V.}
The last step of the procedure is updated following the changes done to~$\Aformula_4^{i,r}$. In particular, for every $i \in [1,o]$ and $r \colon Z \to [m]$, 
if $\Aformula_4^{i,r} = \false$ then $\Aformula_5^{i,r} = \false$, 
otherwise 
\begin{center}
    $\Aformula_5^{i,r}$ \ $\egdef$ \ $m\avar = p_2(t_{2}' - t_{1}') + r_2 + \dots + p_\ell(t_{\ell}' - t_{\ell-1}') + r_\ell + m(c_1 + \dots + c_\ell).$
\end{center}
Let $\Aformula_5^{=} \egdef \bigvee_{i \in [1,o]}\bigvee_{r\colon Z\to [m]} (\Gamma_{i,r} \land \Aformula_5^{i,r})$ be the formula obtained by performing the QE procedure described in this section, on input~$\Aformula_0$. Recall that the formula~$\Gamma_{i,r}$ defined in the second step of the procedure is a conjunction of inequalities with variables from~$\vars{\aformula}$ together with simple modulo constraints.

One can show the following claim with minor adaptation to the proof of correctness of the QE procedure of~\Cref{section:quantifier-elimination}.

\begin{claim}
    $\Aformula_0$ $\fequiv$ $\Aformula_5^=$. The formula $\Aformula_5^=$ is a Boolean combination of linear inequalities and simple modulo constraints.
\end{claim}

\section{Missing proofs from \Cref{subsection:elimination-modulo-quantifiers}}

\ClaimModAFiASix*

\begin{proof}
    Let~$\aformula_{res}$ be the formula~$\bigvee_{s \colon (Z\cup\{\avar\}) \to [mq]} \bigwedge_{\avarfour \in Z \cup \{\avar\}} \avarfour \equiv_{mq} s(\avarfour)$. We have  
    \begin{align}
        \Aformula_5^= &\ \fequiv\ 
                        \aformula_{res} \land \Aformula_5^= \label{afiasix:e0}\\
                    &\ \fequiv\ \aformulater \label{afiasix:e1}\\
                    &\ \fequiv\ 
                    \bigvee_{i \in [1,o]}\bigvee_{s\colon (Z\cup \{\avar\}) \to [mq]} (\Gamma_{i,s} \land \Aformula_5^{i,s})
                    \label{afiasix:e2}
    \end{align}
    Indeed, the equivalence~\eqref{afiasix:e0} holds since~$\aformula_{res}$ is a tautology. The equivalence~\eqref{afiasix:e1}
    holds as $\aformulater$ is obtained from $\aformula_{res} \land \Aformula_5^=$ by distributing the atomic formulae~$\avarfour \equiv_{mq} s(\avarfour)$ of~$\aformula_{res}$ over the disjunctions given by~$\bigvee_{i \in [1,o]}$ and~$\bigvee_{r \colon Z \to [m]}$.
    The nature of~\eqref{afiasix:e2} is already explained during the procedure: since every function $r \colon Z \to [m]$ can be seen as a partial function from~$Z \cup \{\avar\}$ to~$[mq]$, after the two steps above, for every~$s \colon (Z\cup\{\avar\}) \to [mq]$ and~$i \in [1,o]$, all but one disjunct of the subformula $\bigvee_{r \colon Z \to [m]} ((\bigwedge_{\avarfour \in Z \cup \{\avar\}} \avarfour \equiv_{mq} s(\avarfour)) \land \Gamma_{i,r} \land \Aformula_5^{i,r})$ of~$\aformulater$ evaluate $\false$.
    Thanks to~\eqref{afiasix:e0}--\eqref{afiasix:e2}, we conclude: 
    \begin{align}
        \exists \avarter (\avarter \equiv_q \avar \land \Aformula_5^=) &\ \fequiv\ 
        \exists \avarter \Big(\avarter \equiv_q \avar \land \big(\bigvee_{i \in [1,o]}\bigvee_{s\colon (Z\cup \{\avar\}) \to [mq]} (\Gamma_{i,s} \land \Aformula_5^{i,s})\big)\Big)
                    \label{afiasix:e3}\\
        &\ \fequiv\ 
        \Aformula_6^{(\avar,q)} \label{afiasix:e4}
    \end{align}
    The equivalence~\eqref{afiasix:e3} follows from \eqref{afiasix:e0}--\eqref{afiasix:e2}, 
    whereas for the equivalence~\eqref{afiasix:e4} it is sufficient to distribute the existential quantifier $\exists \avarter$ and the formula~$\avarter \equiv_q \avar$ over all the disjunctions given by~$\bigvee_{i \in [1,o]}$ and~$\bigvee_{s \colon (Z\cup\{\avar\}) \to [mq]}$.
\end{proof}

\ClaimModASixASep*

\begin{proof}
    We show that, given $i \in [1,o]$ and $s \colon (Z \cup \{\avar\}) \to [mq]$, 
    ${\exists \avarter (\avarter \equiv_q \avar \land \Gamma_{i,s} \land \Aformula_5^{i,s}) \fequiv \aformulater_{i,s}}$.
    If~$\Aformula_5^{i,s} = \false$ then $\aformulater_{i,s}$ is defined as~$\false$ and the equivalence holds. Otherwise we have 
    $\Aformula_5^{i,s} \egdef m\avarter = p_2(t_{2}' - t_{1}') + r_2 + \dots + p_\ell(t_{\ell}' - t_{\ell-1}') + r_\ell + m(c_1 + \dots + c_\ell)$. 
    In this case, $\aformulater_{i,s} = \Gamma_{i,s}$ if $s(m\avar)$ is congruent to 
    $S \egdef s(p_2(t_{2}' - t_{1}') + r_2 + \dots + p_\ell(t_{\ell}' - t_{\ell-1}') + r_\ell + m(c_1 + \dots + c_\ell))$ modulo~$mq$, and otherwise $\aformulater_{i,s} = \false$. We do a case split following these to cases:
    \begin{itemize}
        \item Suppose that~$s(m\avar)$ and $S$ are not congruent modulo~$mq$. 
        In order to prove that the equivalence ${\exists \avarter (\avarter \equiv_q \avar \land \Gamma_{i,s} \land \Aformula_5^{i,s}) \fequiv \aformulater_{i,s}}$ holds, it is sufficient to show that the formula $\exists \avarter (\avarter \equiv_q \avar \land \Gamma_{i,s} \land \Aformula_5^{i,s})$ is unsatisfiable (since $\aformulater_{i,s}$ is defined as~$\false$). From 
        \begin{center}
            $m\avarter = p_2(t_{2}' - t_{1}') + r_2 + \dots + p_\ell(t_{\ell}' - t_{\ell-1}') + r_\ell + m(c_1 + \dots + c_\ell).$
        \end{center}
        and $\avarter \equiv_q \avar \fequiv m\avarter \equiv_{mq} m\avar$ 
        we derive 
        \begin{center}
            $m\avar \equiv_{mq} p_2(t_{2}' - t_{1}') + r_2 + \dots + p_\ell(t_{\ell}' - t_{\ell-1}') + r_\ell + m(c_1 + \dots + c_\ell).$
        \end{center}
        \emph{Ad absurdum}, suppose that $\aeval \models \exists \avarter (\avarter \equiv_q \avar \land \Gamma_{i,s} \land \Aformula_5^{i,s})$, for some assignment~$\aeval$.
        As~$\aeval \models \Gamma_{i,s}$, for every variable $\avarfour \in Z \cup \{\avar\}$ we have $\aeval(\avarfour) \equiv_{mq} s(\avarfour)$.
        However, this is contradictory, as
        \begin{center}
            $\aeval \models m\avar \equiv_{mq} p_2(t_{2}' - t_{1}') + r_2 + \dots + p_\ell(t_{\ell}' - t_{\ell-1}') + r_\ell + m(c_1 + \dots + c_\ell)$
        \end{center}
        implies $s(m\avar) \equiv_{mq} S$.
        \item Assume that~$s(m\avar)$ and~$S$ are congruent modulo~$mq$. By definition,~$\aformulater_{i,s} = \Gamma_{i,s}$. 
        The left-to-right direction of the equivalence~$\exists \avarter (\avarter \equiv_q \avar \land \Gamma_{i,s} \land \Aformula_5^{i,s}) \fequiv \aformulater_{i,s}$ is thus trivial: if an assignment~$\aeval$ satisfies the left hand side of this equivalence, then $\aeval \models \Gamma_{i,s}$.
        For the right-to-left direction, assume $\aeval \models \Gamma_{i,s}$. 
        Hence, for every variable $\avarfour \in Z \cup \{\avar\}$, 
        $\aeval(\avarfour)$ is equivalent to $s(\avarfour)$ modulo $mq$.
        This implies that $\aeval(m\avar)$ and $V = \aeval(p_2(t_{2}' - t_{1}') + r_2 + \dots + p_\ell(t_{\ell}' - t_{\ell-1}') + r_\ell + m(c_1 + \dots + c_\ell))$ are congruent modulo~$mq$.
        Since moreover $\aeval(m\avar) = m\aeval(\avar)$ is a multiple of $m$, we conclude that there is $v \in \Zed$ such that $m\cdot v = V$.
        Consider the assignment $\aeval\substitute{\avarter}{v}$, with $\avarter$ fresh.
        Since $\avarter \not \in Z \cup \{\avar\}$, we have $\aeval\substitute{\avarter}{v} \models \Gamma_{i,s}$.
        Moreover, from $\aeval(m\avar) \equiv_{mq} mv$, we conclude 
        that $\aeval\substitute{\avarter}{v} \models m\avar \equiv_{mq} m\avarter$. Equivalently, $\aeval\substitute{\avarter}{v} \models \avar \equiv_{q} \avarter$. 
        Lastly, by definition of~$v$,
        \begin{center}
            $\aeval\substitute{\avarter}{v} \models m\avarter = p_2(t_{2}' - t_{1}') + r_2 + \dots + p_\ell(t_{\ell}' - t_{\ell-1}') + r_\ell + m(c_1 + \dots + c_\ell).$
        \end{center}
        Therefore, $\aeval\substitute{\avarter}{v} \models \avarter \equiv_q \avar \land \Gamma_{i,s} \land \Aformula_5^{i,s}$, 
        and thus $\aeval \models \exists \avarter\,(\avarter \equiv_q \avar \land \Gamma_{i,s} \land \Aformula_5^{i,s})$.
        \qedhere
    \end{itemize}
\end{proof}

\begin{restatable}{lemma}{LemmaBoundQuantifierEliminationModulo}
    \label{lemma:bound-quantifier-elimination-modulo}
    The following bounds are established for~$\Aformula_7^{(\avar,q)}$:
    \begin{itemize}
        \setlength{\itemsep}{3pt}
        \item $\fmod{\Aformula_7^{(\avar,q)}} = \{m \cdot q\}$ with $m = k \cdot \lcm{(\fmod{\aformula})}$ and $k \leq \norminf{\homterms{\aformula}}^{\card{\homterms{\aformula}}}$,
        \item $\card{\homterms{\Aformula_7^{(\avar,q)}}} \leq \BigO{\card{\homterms{\aformula}}^2}$ and $\norminf{\homterms{\Aformula_7^{(\avar,q)}}} \leq \BigO{k \cdot \norminf{\homterms{\aformula}}}$,
        \item $\card{\linterms{\Aformula_7^{(\avar,q)}}} \leq \BigO{\card{\linterms{\aformula}}^2}$ and $\norminf{\linterms{\Aformula_7^{(\avar,q)}}} \leq \BigO{k \cdot \norminf{\linterms{\aformula}}}$.
    \end{itemize}
\end{restatable}

\begin{proof}
    The proof follows similarly to the one of~\Cref{lemma:bound-quantifier-elimination-threshold}.
    Without loss of generality, we assume $\card{\homterms{\aformula}}$, $\card{\linterms{\aformula}}$, $\norminf{\homterms{\aformula}}$ and $\norminf{\linterms{\aformula}}$ to be at least $1$. This hides constant factors in the exponent. 
    The quantifier-elimination procedure for~$\exists^{(\avar,q)}\avarbis\,\aformula$ starts by performing the normalisation of the coefficients of~$\avarbis$, 
    as described in Step I of~\Cref{section:quantifier-elimination}. 
    Let $\Aformula_1$ be the resulting formula. 
    As already discussed in the proof of~\Cref{theorem:bound-quantifier-elimination}, 
    we have 
    \begin{itemize}
        \setlength{\itemsep}{3pt}
        \item $\card{\linterms{\Aformula_1}} = \card{\linterms{\aformula}}$ and $\norminf{\linterms{\Aformula_1}} \leq k \cdot \norminf{\linterms{\aformula}}$,
        \item $\card{\homterms{\Aformula_1}} = \card{\homterms{\aformula}}$ and $\norminf{\homterms{\Aformula_1}} \leq k \cdot \norminf{\homterms{\aformula}}$.
    \end{itemize}
    where $k \leq \norminf{\homterms{\aformula}}^{\card{\homterms{\aformula}}}$
    is the \emph{lcm} of all coefficients of~$\avarbis$ appearing in linear inequalities.
    We recall that the formula~$\Aformula_7^{(\avar,q)}$ is defined as $\bigvee_{i \in [1,o]}\bigvee_{s\colon (Z\cup \{\avar\}) \to [mq]} \aformulater_{i,s}$, where ${Z = \freevars{\aformula}}$, $m = \lcm(\fmod{\aformula}) \cdot \norminf{\homterms{\aformula}}^{\card{\homterms{\aformula}}}$, 
    $\aformulater_{i,s}$ is~$\false$ or~$\ordering_i \land (\bigwedge_{\avarfour \in Z\cup\{\avar\}} \avarfour \equiv_{mq} s(\avarfour))$.
    Here, the sub-formula~$\ordering_i$ is an ordering on the set of terms~$T \cup \{0\}$, as defined in Step II of~\Cref{section:quantifier-elimination}. In particular, $\card{T} \leq \card{\linterms{\Aformula_1}}$, $\norminf{T} \leq \norminf{\linterms{\Aformula_1}}$
    and all coefficients of variables in terms of~$T$ are bounded by~$\norminf{\homterms{\Aformula_1}}$.
    So, even when accounting for all orderings~$(\ordering_i)_{i \in [1,o]}$,
    the formula~$\Aformula_7^{(\avar,q)}$ only contains
    at most $(\card{T} \cup \{0\})^2$ inequalities, whose magnitude of coefficients and constants doubles with respect to the one of the terms in~$\Aformula_1$.
    Lastly, every modulo constraint in~$\Aformula_7^{(\avar,q)}$ is of the form~$\avarfour \equiv_{mq} s(\avarfour)$. Therefore, we have the following bounds:
    \begin{itemize}
        \setlength{\itemsep}{3pt}
        \item $\card{\linterms{\Aformula_7^{(\avar,q)}}} \leq (\card{\linterms{\aformula}}+1)^2 \leq 4 \cdot \card{\linterms{\aformula}}^2$,
        \item $\norminf{\linterms{\Aformula_7^{(\avar,q)}}} \leq 2\cdot k \cdot \norminf{\linterms{\aformula}}$,
        \item $\card{\homterms{\Aformula_7^{(\avar,q)}}} \leq (\card{\homterms{\aformula}}+1)^2 \leq 4 \cdot \card{\homterms{\aformula}}^2$,
        \item $\norminf{\homterms{\Aformula_7^{(\avar,q)}}} \leq 2\cdot k \cdot \norminf{\homterms{\aformula}}$,
        \item $\fmod{\Aformula_7^{(\avar,q)}} = \{mq\}$, where $m = k \cdot \lcm(\fmod{\aformula})$,
    \end{itemize}
    where we recall that we are assuming~$\card{\homterms{\aformula}}$, $\card{\linterms{\aformula}}$, $\norminf{\homterms{\aformula}}$ and $\norminf{\linterms{\aformula}}$ to be at least $1$.
\end{proof}

\LemmaBoundQuantifierEliminationModuloSimplified*

\begin{proof}
This is a simple consequence of~\Cref{lemma:bound-quantifier-elimination-modulo}.
\end{proof}

\LemmaBoundQuantifierElimDQUANT*

\begin{proof} 
    Without loss of generality, we assume $\card{\homterms{\aformula}}$, $\card{\linterms{\aformula}}$, $\norminf{\homterms{\aformula}}$ and $\norminf{\linterms{\aformula}}$ to be at least $1$. This hides constant factors in the exponent.
    The proof follows very closely the one of~\Cref{lemma:bound-quantifier-elimination-threshold-d-quant}, with no surprises.
    \begin{itemize}
        \item Let $d$ be the the quantifier-depth of~$\aformula$, 
        \item let $B$ be $2$ plus~$\card{\fmod{\aformula}}$, plus the number of Boolean connectives in~$\aformula$, and 
        \item let $\bar{q}$ be the maximal integer such that~$\exists^{(\avar,\bar{q})}\avarbis$ occurs in~$\aformula$. 
    \end{itemize}
    We show the following bounds for~$\Aformula$, sharpening the ones in the statement of the lemma.

    \begin{itemize}
        \setlength{\itemsep}{3pt}
        \item $\card{\homterms{\Aformula}} \leq A_d \egdef
            (4\cdot B)^{2^d-1} \cdot \card{\homterms{\aformula}}^{2^{d}},$
        \item $\card{\fmod{\Aformula}} \leq B$,
        \item $\norminf{\homterms{\Aformula}} \leq C_d \egdef 2^{(2A_d)^{d}-1}\norminf{\homterms{\aformula}}^{(2A_d)^d}$,
        \item $\norminf{\fmod{\Aformula}} \leq D_d \egdef (\bar{q} \cdot C_d)^{(B^d-1)} \cdot \lcm(\fmod{\aformula})^{B^d}$,
        \item $\card{\linterms{\Aformula}} \leq 
            E_d \egdef (20 \cdot B \cdot {D_d}^2)^{3^d-1} \cdot \card{\linterms{\aformula}}^{3^d}$,
        \item $\norminf{\linterms{\Aformula}} \leq F_d \egdef (6 \cdot {D_d}^3 \cdot E_d)^{d} \cdot \norminf{\linterms{\aformula}}$.
    \end{itemize}
    Recall that the logic features both modulo counting quantifiers and 
    standard first-order quantifiers. To this end,
    notice that~$A_d$, $B$, $C_d$, $E_d$ and $F_d$ overapproximate the homonymous bounds given in~\Cref{lemma:bound-quantifier-elimination-threshold-d-quant} for the threshold quantifiers, for the case where all thresholds~$c$ in~$\exists^{\geq c}\avarbis$ equal~$1$ (since $\exists^{\geq 1} \aformulabis \fequiv \exists \aformulabis$). Therefore, we deal with standard first-order quantifiers exactly as in~\Cref{lemma:bound-quantifier-elimination-threshold-d-quant}. 
    Below, let us focus uniquely on~modulo counting quantifiers.

    Notice that~$A_d$, $C_d$, $D_d$, $E_d$ and $F_d$ are monotonous in $d$.
    Moreover, notice that~$B \leq \BigO{\abs{\aformula}}$.
    The proof is by induction on the quantifier-depth of~$\aformula$ ($d$ takes into account both modulo counting quantifiers and first-order quantifiers). 
    The base case fore $d = 0$, i.e.~$\aformula$ quantifier-free, is trivial.
    For the induction step, let $S = \{\exists^{(\avar_1,q_n)}\avarbis_1\,\aformulabis_1,\dots,\exists^{(\avar_n,q_n)}\avarbis_n\,\aformulabis_n\}$ be a minimal family of formulae such that~$\aformula$ is a Boolean combination of formulae from $S$. Notice that $n \leq B$.
    Let $j \in [1,n]$.
    The quantifier-depth of $\aformulabis_j$ is at most $d-1$.
    We apply the quantifier elimination procedure on~$\aformulabis_j$, obtaining the formula~$\Aformula_j$. By induction hypothesis,
    \begin{itemize}
        \setlength{\itemsep}{3pt}
        \item $\card{\homterms{\Aformula_j}} \leq 
        A_{d-1} = (4\cdot B)^{2^{d-1}-1} \cdot \card{\homterms{\aformula}}^{2^{d-1}},$
        \item $\card{\fmod{\Aformula_j}} \leq B$,
        \item $\norminf{\homterms{\Aformula_j}} \leq C_{d-1} = 2^{(2A_{d-1})^{d-1}-1}\norminf{\homterms{\aformula}}^{(2A_{d-1})^{d-1}}$,
        \item $\norminf{\fmod{\Aformula_j}} \leq D_{d-1} = (\bar{q} \cdot C_{d-1})^{(B^{d-1}-1)} \cdot \lcm(\fmod{\aformula})^{B^{d-1}}$,
        \item $\card{\linterms{\Aformula_j}} \leq E_{d-1} =
        (20 \cdot B \cdot {D_{d-1}}^2)^{(3^{d-1}-1)} \cdot \card{\linterms{\aformula}}^{3^{d-1}}$,
        \item $\norminf{\linterms{\Aformula_j}} \leq F_{d-1} = (6 \cdot {D_{d-1}}^3 \cdot E_{d-1})^{d-1} \cdot \norminf{\linterms{\aformula}}$.
    \end{itemize}
    For~$j \in [1,k]$, we consider every formula $\exists^{(\avar_j,q_j)}\avarbis_j \,\Aformula_j$ and perform 
    the quantifier elimination procedure for threshold counting quantifiers, obtaining a formula $\widetilde{\Aformula}_j$. 
    \Cref{lemma:bound-quantifier-elimination-modulo} (see the proof of this lemma for the exact bounds) 
    we have 
    \begin{itemize}
        \setlength{\itemsep}{3pt}
        \item $\card{\fmod{\widetilde{\Aformula}_j}} = \{q_j \cdot m\}$ with $m = k \cdot \lcm{(\fmod{\Aformula_j})}$ and $k \leq \norminf{\homterms{\Aformula_j}}^{\card{\homterms{\Aformula_j}}}$.
        \item $\card{\linterms{\widetilde{\Aformula}_j}} \leq 4 \cdot \card{\linterms{\Aformula_j}}^2$,
        \item $\norminf{\linterms{\widetilde{\Aformula}_j}} \leq  2 \cdot k \cdot \norminf{\linterms{\Aformula_j}}$,
        \item $\card{\homterms{\widetilde{\Aformula}_j}} \leq 4 \cdot \card{\homterms{\Aformula_j}}^2$
        \item $\norminf{\homterms{\widetilde{\Aformula}_j}} \leq  2 \cdot k \cdot \norminf{\homterms{\Aformula_j}} \leq 2 \cdot \norminf{\homterms{\Aformula_j}}^{2\card{\homterms{\Aformula_j}}}$,
    \end{itemize}
    We derive:
    \begin{itemize}
        \item $\card{\homterms{\widetilde{\Aformula}_j}} \leq 4 \cdot ((4\cdot B)^{2^{d-1}-1} \cdot \card{\homterms{\aformula}})^{2^{d-1}})^2 \leq 
        B^{2^d-2}4^{2^d-1}\card{\homterms{\aformula}}^{2^d} = \frac{A_d}{B}$.
        \item $
        \begin{aligned}[t]
        \norminf{\homterms{\widetilde{\Aformula}_j}} &\leq  2 (2^{(2A_{d-1})^{d-1}-1}\norminf{\homterms{\aformula}}^{(2A_{d-1})^{d-1}})^{2A_{d-1}}\\
        &\leq 2^{(2A_{d-1})^{d}-1}\norminf{\homterms{\aformula}}^{(2A_{d-1})^d} \leq 2^{(2A_{d})^{d}-1}\norminf{\homterms{\aformula}}^{(2A_{d})^d} = C_d.
        \end{aligned}$
    \end{itemize}
    Notice that $k \leq \norminf{\homterms{\Psi_j}}^{\card{\homterms{\Psi_j}}} \leq C_d$ and that $\lcm(\fmod{\Aformula_j}) \leq \norminf{\fmod{\Aformula_j}}^B$.
    \begin{itemize}
        \item $
        \begin{aligned}[t]
        m &\leq  q_j \cdot C_d \cdot \lcm{(\fmod{\Aformula_j})} \leq \bar{q} \cdot C_d \cdot \norminf{\fmod{\Aformula_j}}^B\\
        &\leq \bar{q} \cdot C_d \cdot ((\bar{q} \cdot C_{d-1})^{(B^{d-1}-1)} \cdot \lcm(\fmod{\aformula})^{B^{d-1}})^B\\
        &\leq (\bar{q} \cdot C_d)^{(B^d-B+1)} \cdot \lcm(\fmod{\aformula})^{B^d}
        \leq (\bar{q} \cdot C_d)^{(B^d-1)} \cdot \lcm(\fmod{\aformula})^{B^d} = D_d,
        \end{aligned}$
        \item[] 
        where we recall that we assume $B \geq 2$.
        Hence, $\card{\fmod{\widetilde{\Aformula}_j}} = 1$ and $\norminf{\fmod{\widetilde{\Aformula}}_j} \leq D_d$.
        \item $
            \card{\linterms{\widetilde{\Aformula}_j}} \leq 4 \cdot ((20 \cdot B \cdot {D_{d-1}}^2)^{3^{d-1}-1} \card{\linterms{\aformula}})^{3^{d-1}})^2 \leq 
            B^{3^d-2} (20 \cdot {D_d}^2)^{3^d-1}\card{\linterms{\aformula}}^{3^d} \leq \frac{E_d}{B}.$
        \item $\norminf{\linterms{\widetilde{\Aformula}_j}} \leq  2 \cdot C_d \cdot \norminf{\linterms{\Psi_j}} \leq (2 \cdot C_d) \cdot (6 \cdot {D_{d-1}}^3 \cdot E_{d-1})^{d-1} \cdot \norminf{\linterms{\aformula}} \leq F_d$, 
        
        where we notice that~$C_d \leq D_d$.
    \end{itemize}
    For every $j \in [1,n]$, we replace in $\aformula$ 
    each occurrence of the formula~$\exists^{(\avar_j,q_j)}\avar_j\,\aformulabis_j$ 
    not in the scope of a quantification with the formula~$\widetilde{\Aformula}_j$. We obtain the formula~$\Aformula$ that is a Boolean combination of 
    at most $B$ formulae from $\{\widetilde{\Aformula}_1,\dots,\widetilde{\Aformula}_n\}$. So, $\Aformula$ is quantifier-free.
    We have:
    \begin{itemize}
        \item $\card{\homterms{\Aformula}} = \sum_{j = 1}^n \card{\homterms{\widetilde{\Aformula}_j}} \leq A_d$,
        \item $\card{\fmod{\Aformula}} = \sum_{j = 1}^n \card{\fmod{\widetilde{\Aformula}_j}} \leq B$, 
        \item 
        $\norminf{\homterms{\Aformula}} \leq \max\{\norminf{\homterms{\widetilde{\Aformula}_j}}\mid j \in [1,n]\} \leq C_d$,
        \item 
        ${\norminf{\fmod{\Aformula}} \leq \max\{\norminf{\fmod{\widetilde{\Aformula}_j}}\mid j \in [1,n]\} \leq D_d}$,
        \item
        $\card{\linterms{\Aformula}} = \sum_{j = 1}^n \card{\linterms{\widetilde{\Aformula}_j}} \leq E_d$,
        \item 
        $\norminf{\linterms{\Aformula}} \leq \max\{\norminf{\linterms{\widetilde{\Aformula}_j}}\mid j \in [1,n]\} \leq F_d$.
        \qedhere
    \end{itemize}
\end{proof}

\end{document}